\documentclass[12pt]{article}
\usepackage{amsmath}
\usepackage{graphicx,psfrag,epsf}
\usepackage{enumerate}
\usepackage{url} 

\newcommand{\blind}{1}

\usepackage[english]{babel}
\usepackage[utf8]{inputenc}
\usepackage[T1]{fontenc}
\usepackage{lmodern}
\usepackage{amssymb}
\usepackage{amsthm}
\usepackage{mathtools}
\usepackage{float}
\usepackage{bm}

\usepackage{adjustbox}
\usepackage{booktabs}
\usepackage{multirow}

\DeclareMathOperator*{\argmin}{arg\,min}

\newtheoremstyle{mystyle}
  {}
  {}
  {}
  {}
  {\itshape}
  {.}
  { }
  {}%
\theoremstyle{mystyle}

\newtheorem{proposition}{Proposition}[subsection]

\newtheorem{definition}{Definition}[subsection]

\newtheorem*{assumption*}{\assumptionnumber}
\providecommand{\assumptionnumber}{}


\newcommand{\norm}[1]{\lVert#1\rVert}

\usepackage[font=footnotesize,skip=-1pt,justification=centering]{caption}

\let\originalleft\left
\let\originalright\right
\renewcommand{\left}{\mathopen{}\mathclose\bgroup\originalleft}
\renewcommand{\right}{\aftergroup\egroup\originalright}

\usepackage{titlesec}
\titleformat{\subsubsection}[runin]
{\normalfont\itshape}{\thesubsubsection}{0.5em}{}
\titlespacing{\subsubsection}{\parindent}{3.25ex plus 1ex minus .2ex}{1.5ex plus .2ex}

\usepackage[authoryear]{natbib}
\setcitestyle{citesep={;}}
\setcitestyle{aysep={}} 

\let\OLDthebibliography\thebibliography
\renewcommand\thebibliography[1]{
  \OLDthebibliography{#1}
  \setlength{\parskip}{0pt}
  \setlength{\itemsep}{0pt plus 0.3ex}
}

\usepackage[linktoc=all,colorlinks=true,linkcolor=black,citecolor=blue,urlcolor=blue,bookmarks=false,hypertexnames=true]{hyperref}

\begin{document}

    \date{}
    \def\spacingset#1{\renewcommand{\baselinestretch}%
    {#1}\small\normalsize} \spacingset{1}
    
    \if1\blind
    {
      \title{\bf Simultaneous Feature Selection and Outlier
        Detection with Optimality Guarantees}
      \author{Luca Insolia\thanks{
            This work was partially funded by the NIH B2D2K training grant and the Huck Institutes of the Life Sciences of Penn State. Computation was performed on the Institute for Computational and Data Sciences Advanced CyberInfrastructure (ICDS-ACI; Penn State University).}
        \hspace{.2cm}\\
        \normalsize Faculty of Sciences, Scuola Normale Superiore 
        \vspace{.2cm} \\
        Ana Kenney\\
        \normalsize Dept.~of Statistics, Penn State University
        \vspace{.2cm} \\
        Francesca Chiaromonte\\
        \normalsize Dept.~of Statistics, Penn State University \\
        \normalsize Inst.~of Economics \& EMbeDS, Sant'Anna School of Advanced Studies 
        \vspace{.2cm} \\
        Giovanni Felici\\
        \normalsize Istituto di Analisi dei Sistemi ed Informatica, Consiglio Nazionale delle Ricerche}
      \maketitle
    } \fi
    
    \if0\blind
    {
      \bigskip
      \bigskip
      \bigskip
      \begin{center}
        {\LARGE\bf Simultaneous Feature Selection and Outlier\\ \vskip0.2cm
        Detection with Optimality Guarantees}
    \end{center}
      \medskip
    } \fi
    
    \begin{abstract}
        \noindent
        Sparse estimation methods capable of tolerating outliers have been broadly investigated in the last decade.
        We contribute to this research considering high-dimensional regression problems contaminated by multiple \textit{mean-shift outliers} which affect both the response and the design matrix.
        We develop a general framework for this class of problems and propose the use of \textit{mixed-integer programming} to simultaneously perform feature selection and outlier detection with provably optimal guarantees. 
        We characterize the theoretical properties of our approach, i.e.~a necessary and sufficient condition for the \textit{robustly strong oracle property}, which allows the number of features to exponentially increase with the sample size; the optimal  estimation of the parameters; and
        the breakdown point of the resulting estimates.
        Moreover, we provide computationally efficient procedures to tune integer constraints and to warm-start the algorithm.
        We show the superior performance of our proposal compared to existing heuristic methods through numerical simulations and an application investigating the relationships between the human microbiome and childhood obesity. 
    \end{abstract}
    
    
    \spacingset{1} 
    
    \section{Introduction}
    \label{sec:intro}

    	High-dimensional regression problems have become ubiquitous in most application domains. In these problems the number of features recorded on each observation (or case) is very large -- possibly larger than the sample size, and often growing with the sample size itself.
    	The availability of ever larger numbers of potential predictors increases both the chances that some substantial portion of them are irrelevant, and the chances of contamination in the data (i.e.~some cases following a different model).
    	Thus, it is critical to employ estimation approaches that address both {\em sparsity} and {\em statistical robustness}. 
        The \textit{Mean-Shift Outlier Model} (MSOM) is a common paradigm leading to the exclusion of outliers from the fit \citep{beckman1983outlier}.
    	In the high-dimensional setting, traditional approaches focused on robustifying information criteria or resampling methods \citep{muller2005outlier}.
    	In the last decade several \textit{robust penalization methods} have been introduced in the literature.
    	They generally rely on a robustification of soft-selection procedures \citep{alfons2013sparse}, adopting a case-wise robust counterpart of \textit{Maximum Likelihood Estimation} (MLE). 
        The MSOM can be equivalently parametrized with the inclusion of binary variables \citep{cook1982residuals}, transforming outlier detection into a feature selection problem \citep{morgenthaler2004algorithms} and motivating the development of methods for \textit{simultaneous feature selection and outlier detection} (SFSOD).
    
        In contrast to existing methodologies which extensively rely on heuristics, we propose a discrete and provably optimal approach to perform SFSOD, highlighting its connections with other approaches.
        The $L_0$ constraint has been used separately in the context of feature selection (\citealt{bertsimas2016best}; \citealt{bertsimas2020sparse}) and robust estimation (\citealt{zioutas2009quadratic}; \citealt{bertsimas2014least}) -- both of which can be formulated as a \textit{Mixed-Integer Program} (MIP) and solved with optimality guarantees. 
        We combine these two approaches into a novel formulation and take advantage of existing heuristics to provide effective big-$\mathcal{M}$ bounds and warm-starts to reduce the computational burden of MIP. 
    	We provide theoretical guarantees for our approach, including its high breakdown point, a necessary and sufficient condition to achieve the \textit{robustly strong oracle property} -- which holds also in the ultra-high dimensional case with the number of features exponentially increasing with the sample size -- and optimal parameter estimation. In contrast to existing methods, our approach requires weaker assumptions and allows both the feature sparsity level and the amount of contamination to depend on the number of predictors and the sample size, respectively.
    	Moreover, we propose criteria to tune, in a computationally efficient and data-driven way, both the sparsity of the solution and the estimated amount of contamination.
    
        The article is organized as follows: 
        Section~\ref{sec1:background} discusses the relevant background for outlier detection and feature selection.
        Section~\ref{sec2:proposal} details our proposal -- including a general framework for SFSOD, the MIP formulation and its theoretical properties. 
        Section~\ref{sec3:sim} presents a simulation study comparing our proposal with state-of-the-art methods and Section~\ref{sec4:appl} presents our real data application investigating the relationships between child obesity and the human microbiome.
        Final remarks are included in Section~\ref{sec5:final} and additional details are provided in the Supplementary Material.

        \section{Background} \label{sec1:background}
        
            Consider a regression model of the form $  \bm{y} = \bm{X} \bm{\beta} + \bm{\varepsilon} $, where $ \bm{y} \in \mathbb{R}^{n} $ is the response vector, $ \bm{\varepsilon} \in \mathbb{R}^{n}$ the error vector with a $N( \bm{0}, \sigma^2 \bm{I}_n )$ distribution ($\bm{I}_n$ is the identity matrix of size $n$), $ \bm{X} \in \mathbb{R}^{n \times p} $ the design matrix, and $ \bm{\beta} \in \mathbb{R}^{p}$ the vector of regression coefficients.
    		In the following, we describe the relevant background for our proposal.
    		In particular, we briefly review methods for outlier detection, and present the equivalent formulation as a feature selection problem.
    		We then discuss approaches for model selection, focusing on the use of an $L_0$ constraint for best subset selection.

    		We consider a case-wise contamination mechanism, where each outlying unit might be contaminated in some (or even all) of its dimensions. In particular, we assume that outliers are generated by an MSOM, where the set of outliers $ M = \{ i \in \{ 1, \ldots, n \}: \varepsilon_i \sim N(\lambda_{\varepsilon_i}, \sigma^2), \lambda_{\varepsilon_i} \neq 0  \}$ has cardinality $\lvert M \rvert = n_0$.
   	        For a given dimension $p < n-n_0$, it is well-known that MLE leads to the removal of outliers from the fit \citep{cook1982residuals}. 
       	    Moreover, as is customary, we assume that the MSOM can also affect the design matrix $\bm{X}$ with mean shifts $\lambda_{\bm{x}_i}$ \citep{maronna2006robust}.
	        
	        If a regression comprises a single outlier, its position corresponds to the unit with largest absolute Studentized residual, which is a monotone transformation of the deletion residual 
	        $
        	    t_i = (y_i- \bm{x}_i^{T} \hat{\bm{\beta}}_{(i)}) / (\hat{\sigma}_{(i)} 
        	    (1+ \bm{x}_{i}^{T}( \bm{X}_{(i)}^{T} \bm{X}_{(i)})^{-1} \bm{x}_i)^{1 / 2} ,
    	    $
	        where the subscript $(i)$ indicates the removal of the $i$-th unit \citep{atkinson1985plots}.
	        Under the null model, $t_i$ is distributed as a Student's $t$ with $n-p-1$ degrees of freedom, which can be computed from an MLE fit based on all units and used as a test for outlying-ness of single data points \citep{cook1982residuals}.
    	    The same idea can be easily generalized to regressions with multiple outliers.
    	    Operationally though, this was considered ineffective -- due to the high likelihood of masking (undetected outlying cases) and swamping (non-outlying cases flagged as outliers) effects \citep{huber2009robust} -- and computationally intractable \citep{bernholt2006robust}. 
    	    This motivated the development of high-breakdown point estimators such as the Least Trimmed Squares (LTS), S, and MM (\citealt{maronna2006robust}, see Section~\ref{subsec:theory} for more details); 
            outlier detection and high-breakdown point estimation are historically distinct but closely related areas of statistical research \citep{huber2009robust}.
    		
            Assuming without loss of generality that outliers occupy the first $n_0$ positions in  the data, the MSOM can be equivalently parametrized as
    		$
    		    \bm{y} = \bm{X} \bm{\beta} + \bm{D}_{n_0} \bm{\phi} + \bm{\varepsilon} , 
		    $
    	    where the original design matrix $\bm{X}$ is augmented with a binary matrix $\bm{D}_{n_0} = [ \bm{I}_{n_0}, \bm{0}]^T$ of size $n \times n_0 $ indexing the $n_0$ outliers \citep{morgenthaler2004algorithms}.
    	    If $p < n-n_0$, the MLE for $\bm{\phi} \in \mathbb{R}^{n_0} $ provides prediction residuals for the $n_0$ units excluded from the fit; i.e.~their residuals under a model which excludes them from the estimation process. This is given by
    	    $
    			\hat{\bm{\phi}}  = \left[ {\bm{I}_{n_0} - \bm{H}_{MM}} \right]^{-1} (\bm{y}_{M}- \bm{X}_{M}^{T} \hat{\bm{\beta}}_{(M)}) ,
			$
    	    where $ \bm{H}_{MM} = \bm{X}_M (\bm{X}^T\bm{X})^{-1} \bm{X}_M^T$, and the associated $t$-statistics $ \bm{t}_{M} $ provide (multiple) deletion residuals.
        	However, masking and swamping effects can again arise if $\bm{D}_{n_0}$ does not index all possible outliers. 
        	
        	Outlier detection in low-dimensional problems can be performed substituting the identity matrix $\bm{I}_n$ in place of $\bm{D}_{n_0}$ and applying feature selection methods to $\bm{\phi} \in \mathbb{R}^{n} $ to identify outlying cases. 
        	The literature contains examples of both convex 
        	\citep{mccann2006robust,taylan2014approach,liu2019minimizing}
            and non-convex 
            \citep{she2011outlier,liu2017robust,gomez2019outlier,barratt2020minimizing} penalization methods applied to this problem; notably, the latter are necessary to achieve high-breakdown point estimates.
        	
        	
        	Penalization methods are also the hallmark of feature selection  in high dimensional problems, where they seek to induce sparsity estimating $p_0 < p$ non-zero coefficients in $ \bm{\beta}$ -- whose dimension $p$ can exceed $n$.
        	Soft penalization methods such as Lasso \citep{tibshirani1996regression}, SCAD \citep{fan2001variable}, adaptive Lasso \citep{zou2006adaptive} and MCP \citep{zhang2010nearly} rely on non-differentiable  continuous penalties, which can be convex or non-convex. 
        	They can be formulated as
        	$
    		\hat{\bm{\beta}}  = \operatorname*{\arg\min}_{\bm{\beta}} ~
        		\norm{ \bm{y} - \bm{X} \bm{\beta}}_2^2 + R_\omega(\bm{\beta}) 
        	$, 
        	where the penalty function $R_\omega(\bm{\beta}) $ depends on a tuning parameter $\omega$.
        	
        	Best subset selection, a traditional hard penalization method, solves feature selection combinatorially, comparing all possible models of size $p_0$ \citep{miller2002subset}.
        	This can be formulated as a MIP through an $L_0$ constraint on $\bm{\beta}$, where the $L_0$ pseudo-norm is defined as 
        	$ \norm{\bm{\beta}}_0 = \sum_j I(\beta_j \neq 0) $ ($I(\cdot)$ is the indicator function).
        	The MIP formulation of best subset selection is computationally intractable \citep{natarajan1995sparse} and was previously considered impossible to solve with optimality guarantees for regression problems with $p$ larger than a few tens \citep{hastie2015statistical}.
            Nevertheless, improvements in optimization solvers and hardware components, which experienced a 450 billion factor speed-up between 1991 and 2015 alone, now allow one to efficiently solve problems of realistic size with provable optimality \citep{bertsimas2016best}.
            Modern MIP solvers are based on implicit enumeration methods along with constraints such as \textit{cutting planes} that tighten the relaxed problem
            (\textit{Branch \& Bound} and  \textit{Branch \& Cut}, \citealt{Schrijver1986}).
            Optimality is certified monitoring the gap between the best feasible solution and the relaxation of the problem.
    		Importantly, MIP methods can recover the subset of true active features (i.e.~they satisfy oracle properties, see Section~\ref{subsec:theory}) under weaker assumptions compared to soft penalization methods.

		\section{Proposed methodology}
		\label{sec2:proposal}
    		
    		We focus on a regression comprising both outliers and inactive features, where one has to tackle at the same time an {\em unlabeled} MSOM problem (i.e.~one where the identity, number and strength of outliers are unknown, \citealt{beckman1983outlier}) and the sparse estimation of $\bm{\beta}$.
    		SFSOD can be framed as an optimization problem; namely:
    		\begin{align}\label{eqSODFS}
        		\left[ \hat{\bm{\beta}}, \hat{\bm{\phi}} \right]  = \operatorname*{\arg\min}_{\bm{\beta}, \bm{\phi}} ~
        		& \sum_{i=1}^n \mathcal{\rho} (y_i,  f(\bm{x}_i; \bm{\beta}) + \phi_i ) \\
        		\text{s.t.} \quad & R_\omega(\bm{\beta}) \leq c_\beta \nonumber\\
        		& ~ R_\gamma( \bm{\phi} ) \leq c_\phi \nonumber,
    		\end{align}
    		where $\rho(\cdot)$ is a loss function, 
    		$f(\cdot)$ defines the relation between predictors and response vector,
    		and $R_\omega(\bm{\beta})$ and $R_\gamma(\bm{\phi})$ are penalties subject to sparsity-inducing constraints, which may depend on tuning constants $\omega$ and $\gamma$. 
    		Non-zero coefficients in $\hat{\bm{\beta}}$ and $\hat{\bm{\phi}}$ identify active features and outlying units, respectively. 
    	    Although in this article we focus on linear regression the framework in \eqref{eqSODFS} is very general; it comprises Generalized Linear Models, several classification techniques and non-parametric methods. 
    		
    		Many approaches have been recently developed to solve \eqref{eqSODFS} using \textit{Ordinary Least Squares} (OLS) as the loss function $\rho(\cdot)$.
            Both penalties $ R_\omega(\beta) $ and $ R_\gamma( \phi ) $ are generally convex
            \citep{morgenthaler2004algorithms,menjoge2010diagnostic,lee2012regularization,kong2018fully} 
            although some non-convex procedures have been considered \citep{she2011outlier}. 
            Similar ideas have been explored from a Bayesian perspective \citep{hoeting1996method} 
            and for the estimation of units' weights \citep{li2012simultaneous, xiong2013regression}.
            Robust soft penalization methods also can be cast into \eqref{eqSODFS}, abandoning the explicit use of $\bm{\phi}$ and adopting a robust loss $\rho(\cdot) $ in place of the OLS.
            These include MM-estimators for ridge regression \citep{maronna2011robust}, sparse-LTS \citep{alfons2013sparse}, bridge MM-estimators \citep{smucler2017robust}, enet-LTS \citep{kurnaz2017robust}, penalized elastic net S-estimators \citep{freue2019robust}, and penalized M-estimators \citep{loh2017statistical,chang2018robust,amato2020penalised}, as well as their re-weighted counterparts.
    		Indeed, through specific penalties, M-estimators can be equivalently formulated as feature selection problems \citep{she2011outlier}.
    		Related approaches include the robust LARS \citep{khan2007robust}, robust adaptive Lasso \citep{machkour2017outlier}, and penalized exponential squared loss \citep{wang2013robust}.
    		
            While \eqref{eqSODFS} highlights an important parallel between SFSOD and robust soft penalization, existing heuristic methods suffer several drawbacks.
            Some rely on restrictive assumptions or their finite-sample and asymptotic performance in terms of feature selection and outlier detection is not well-established.
            Others rely heavily on an initial subset of cases identified as non-outlying. 
            Yet others provide a down-weighting of all units, which complicates interpretation and the objective identification of outliers, or have an asymptotic breakdown point of $0\%$, so they in fact do not tolerate outliers in the first place.
            Finally, some methods require tuning of several parameters in addition to $\omega$ and $\gamma$, which can severely increase  computational burden.
            Our proposal solves \eqref{eqSODFS} with optimality guarantees, from both optimization and theoretical perspectives. 
            This preserves the intrinsic discreteness of the problem, facilitating implementation, interpretation, and generalizations.

		\subsection{MIP formulation}
		    
    	    We impose two separate integer constraints on $\bm{\beta}$ and $\bm{\phi}$ in \eqref{eqSODFS}, combining in a single framework the use of $L_0$ constraints for feature selection \citep{bertsimas2016best,kenney2018efficient,bertsimas2020sparse} and outlier detection \citep{bertsimas2014least,zioutas2009quadratic}.
    	    In particular, we consider the following MIP formulation: 
    	    \begin{subequations}\label{eq:miqp}
        	    \begin{align}
                	\left[ \hat{\bm{\beta}}, \hat{\bm{\phi}} \right]  = 
                	\argmin_{\bm{\beta},\bm{z}^{\beta},\bm{\phi},\bm{z}^{\phi}}~
                    &\frac{1}{n} \rho(\bm{y} - \bm{X} \bm{\beta} - \bm{\phi}) \tag{\ref{eq:miqp}} \\
                    \mathrm{s.t.} \quad &-\mathcal{M}_j^{\beta} z^{\beta}_j \leq  \beta_j \leq \mathcal{M}_j^{\beta} z^{\beta}_j   \label{eq:miqp_bigMbeta}\\
                    &-\mathcal{M}_i^{\phi} z^{\phi}_i \leq  \phi_i \leq \mathcal{M}_i^{\phi}z^{\phi}_i  \label{eq:miqp_bigMphi} \\
                    &\sum_{j=1}^p \beta_j^2 \leq \lambda  \label{eq:mip_ridge} \\
                    &\sum_{j=1}^p z^{\beta}_j\leq k_p \label{eq:mip_intB} \\
                    &\sum_{i=1}^n z^{\phi}_i\leq k_n \label{eq:mip_intP} \\
                    &\beta_j \in \mathbb{R}, \ \ z^{\beta}_j \in \{0,1\}  , \qquad j=1,\ldots, p \label{eq:mip_B} \\
                    &\phi_i \in \mathbb{R}  , \ \ z^{\phi}_i \in \{0,1\}  , \qquad i=1,\ldots, n , \label{eq:mip_P}
                \end{align}
            \end{subequations}
            where $\bm{\mathcal{M}}^{\beta}$ and $\bm{\mathcal{M}}^{\phi}$ in constraints (2a) and (2b) are the so-called big-$\mathcal{M}$ bounds \citep{Schrijver1986}.
            In our proposal these are vectors of lengths $p$ and $n$, respectively, which can be tailored for each $\beta_j$ and $\phi_i$.
            In the $L_0$-norm constraints~\eqref{eq:mip_intB} and~\eqref{eq:mip_intP}, $k_p$ and $k_n$ are positive integers modulating sparsity for feature selection and outlier detection, respectively -- for the latter, we can think of sparsity as a level of trimming (i.e.~outlier removal).
            In the $L_2$-norm ridge-like constraint~\eqref{eq:mip_ridge}, $\lambda$ is a positive scalar that can be used to counteract strong collinearities among the features \citep{hoerl1970bridge}.
            Here it also modulates a trade-off between continuity and unbiasedness in the estimation of $\bm{\beta}$, and allows one to calibrate the intrinsic discreteness of the problem 
            -- making its solutions more stable with respect to small data perturbations \citep{breiman1995better} and weak signal-to-noise ratio regimes \citep{hastie2017extended}.
            
            Although solving \eqref{eq:miqp} plainly with state-of-the-art software may be computationally intractable for large dimensions,
            with the appropriate implementation it can be used to tackle many real-world applications optimally and efficiently.
            Another important advantage of our framework from an application standpoint is that it allows one to easily incorporate additional constraints to leverage structure in the data -- such as groups \citep{zou2005regularization}, ranked features \citep{tibshirani2005fused}, hierarchical interactions \citep{bien2013lasso} and compositional information \citep{lin2014variable}.
        	Moreover, the $L_0$-penalty (also called entropy penalty) is 
        	not equivalent to the $L_0$ constrained formulation used in MIP due to non-convexity
        	\citep{shen2013constrained}.
            
            \subsection {Some implementation details}
            \label{sec2:proposalTech}
            
            As customary in feature selection problems, we do  not penalize the intercept and we standardize features at the outset. Because the regressions we focus on comprise outliers, we use a robust standardization; 
            both $\bm{y}$ and $\bm{X}$ are centered to have zero medians, and each $\bm{X}_j$ is also scaled to have unit median absolute deviation (MAD); results in the output are given in the original scale.
    	    Although binary features are often standardized \citep{tibshirani1997lasso}, since the constraints on $\bm{\beta}$ and $\bm{\phi}$ are separate, we do not standardize the binary variables used in the outlier detection component of the problem.
    	    Note that the interpretation of the entries in $\bm{\phi}$ as prediction residuals indicates that they are already on the same scale under the null model.

            Setting the big-$\mathcal{M}$ bounds for \eqref{eq:miqp} is made even more complicated due to the ``double'' nature of SFSOD.
            A robust estimator of the regression coefficients, say $\tilde{\bm{\beta}}$, can be used to set $\bm{\mathcal{M}}^{\beta} = \tilde{\bm{\beta}} c$ and 
            $ \bm{\mathcal{M}}^{\phi} = (\bm{y} - \bm{X} \tilde{\bm{\beta}}) c = \tilde{\bm{e}} c $,
            where $c \geq 1$ is a suitable multiplicative constant.
            We generalize this approach using an \textit{ensemble} $\tilde{\bm{\beta}}_t $ (for $t=1,\ldots,T$) of preliminary estimators and setting 
            $    \mathcal{M}^\beta_j = \max_t ( \lvert \tilde{\beta}_{t_j} \rvert ) c
            $
            and 
            $  \mathcal{M}^\phi_i = \max_t ( \lvert \tilde{e}_{t_j}  \rvert ) c 
            $.
            The ensemble guarantees that, if at least one of the 
            $\tilde{\bm{\beta}}_{t}$'s is reasonably close to the optimal solution, the MIP will easily recover such solution.
            Importantly, having also non-robust or non-sparse estimators in the ensemble does not negatively affect solution quality but only convergence speed.
                
            The MIP formulation in \eqref{eq:miqp} critically depends on the big-$\mathcal{M}$ bounds; 
            they should be large enough to retain the optimal solution, yet small enough to avoid unnecessary computations and numerical instability.
            If identifying suitable bounds is not possible, we use an alternative strategy based on \textit{Specially Ordered Sets of Type 1} (SOS-1)  \citep{bertsimas2016best}. 
            These allow only one variable in the set to be non-zero,
            e.g.,~$(1-z^{\beta}_j, \beta_j) = 0 \iff (1-z^{\beta}_j, \beta_j) : \text{SOS-1} $.
            Hence, constraints \eqref{eq:miqp_bigMbeta} and \eqref{eq:miqp_bigMphi} can be equivalently formulated as:	  
    	    \begin{subequations}
        	    \begin{align}
                    & (1-z^{\beta}_j, \beta_j) : \text{SOS-1}, \qquad j=1,\ldots, p  \label{eq:sosp}\\
                    & (1-z^{\phi}_i, \phi_i) : \text{SOS-1}, \qquad i=1,\ldots, n, \label{eq:sosn}
                \end{align}
            \end{subequations}
            which can be solved via modern MIP solvers such as \texttt{Gurobi} or \texttt{CPLEX}.
            SOS-1 constraints in \eqref{eq:miqp} guarantee that the global optimum can be reached, and generally outperform  big-$\mathcal{M}$ bounds when these are difficult to reasonably set.

	        The formulation in \eqref{eq:miqp} also, and critically, requires the tuning of $k_p$, $k_n$ and, if a ridge-like constraint is included in the model, $\lambda$. 
	        Performing this simultaneously along an extensive grid of values can be computationally unviable for MIP. 
	        We therefore proceed as follows: 
	        {\bf (i)} fix $\lambda$ (possibly, in turn, to a few values in a meaningful range);
            {\bf (ii)} fix $k_n$ to a starting value larger than a reasonable expectation on the amount of contamination in the problem ($n_0$);
            {\bf (iii)} holding fixed the $k_n$ starting value from (ii), tune $k_p$ through cross-validation or an information criterion; 
            {\bf (iv)} holding fixed the $k_p$ value selected in (iii), refine downward the value of $k_n$. 
            
            For cross-validation in (iii) we use the  computationally efficient integrated scheme introduced in \citet{kenney2018efficient}, robustifying the performance measure (the mean squared prediction error) with an upper trimmed sum.
            Choosing the trimming level is again not trivial, because cross-validation folds might contain different proportions of outliers.
            In order to be conservative, we fix the trimming proportion to $3 k_n / n $ on the test fold and to $ 2 k_n / n  $ on the training folds.
            For information criteria in (iii) the situation is more straightforward, as one can compute robust values for them using only cases identified as non-outlying in any given MIP run.
	        Refining $k_n$ downward in (iv) improves efficiency in estimating $\bm{\beta}$, which can be low if the starting $k_n$ in (ii) is substantially larger than the true $n_0$, excluding non-outlying cases from the fit.
	        Assuming that the selected model of size $k_p$ in (iii) is close to the true $p_0$ active features, iteratively reducing $k_n$ provides an effective strategy to pinpoint when outliers start to be included in the fit. 
	        Similarly to the Forward Search algorithm \citep{atkinson2000fs}, this can be done monitoring an appropriate statistic (e.g.,~the minimum absolute deletion residuals) along iterations.

      \subsection{Theoretical results}
        \label{subsec:theory}
        
        In this Section, we characterize the theoretical properties of our proposal through two groups of results. 
        The first comprises properties established under the general framework introduced in \eqref{eq:miqp}.
        The second comprises key properties established under an $L_2$-norm loss function ($\rho(\cdot)= \norm{\cdot}^2_2$); namely, the \textit{robustly strong oracle property} and {\em optimal parameter estimation} for SFSOD. 
        All proofs are provided in the Supplementary Material.
        
        Without much loss of generality, we assume that \eqref{eq:miqp} has a unique global minimum, and that the loss function is such that 
        $\rho( \bm{x} ) \geq 0 $ with $\rho(\bm{0}) = 0$ (this is the case for OLS and many other instances, such as estimation in quantile regression 
        and robust estimators).
        Our first result connects our proposal to a large class of penalized methods based on trimming.
        \begin{proposition}[Sparse trimming] \label{lemma:1}
            For any $\lambda$, $n$, $p$, $k_n$ and $k_p$,
            the $\hat{\bm{\beta}}$ estimator produced solving \eqref{eq:miqp} is the same as the one  produced solving 
            \begin{align} \label{eq:quival_lts}
                 \argmin_{\bm{\beta}} ~ & \frac{1}{n} \sum_{i=1}^{n-k_n} (\rho ( y_i - \bm{x}_i^T \bm{\beta} ) )_{i:n} 
                =
                \frac{1}{n} \sum_{i=1}^{n-k_n} ( \rho ( e_i ))_{i:n}  \\
                \mathrm{s.t.} \quad & \eqref{eq:miqp_bigMbeta}, \eqref{eq:mip_ridge}, \eqref{eq:mip_intB},  \eqref{eq:mip_B} \nonumber
            \end{align}
            where $e_i$ (for $i=1,\ldots,n$) are the residuals, and 
            $ (\rho ( e_1 ))_{1:n} \leq \ldots \leq (\rho ( e_n ))_{n:n} $ the order statistics of their $\rho(\cdot)$ transformation. 
        \end{proposition}
        \noindent
        Proposition~\ref{lemma:1} demonstrates the equivalence of our formulation to a trimmed loss problem, where the level of trimming is directly controlled by the $L_0$ constraint on $\bm{\phi}$. This extends a well-known result for unpenalized OLS and motivates the formulation in \eqref{eqSODFS} as a general framework for SFSOD.
        In particular, \eqref{eq:quival_lts} includes as special cases the LTS and the Trimmed Likelihood Estimator (\citealt{rousseeuw1984robust}; \citealt{hadi1997maximum}).
        Thus, our MIP approach for SFSOD inherits the desirable properties of these robust estimators.

        The largest proportion of outliers that an estimator can tolerate before becoming arbitrarily biased is referred to as the \textit{breakdown point}. 
        In symbols, consider a sample $\bm{Z} = (\bm{z}_1, \ldots, \bm{z}_n) $ with $ \bm{z}_i = (y_i, \bm{x}_i^T) $.
        The \textit{maximum bias} for an estimator, say $\bm{\tau}$, is
		$ b^*(n_0 ; \bm{\tau}, \bm{Z})=\sup_{\tilde{\bm{Z}}}  \norm{\bm{\tau}(\tilde{\bm{Z}})-\bm{\tau}(\bm{Z})}_2 $,
		where $\tilde{\bm{Z}}$ represents $\bm{Z}$ after the replacement of $n_0$ points by arbitrary values.
		The \textit{finite-sample replacement breakdown point} (BdP henceforth), defined as
		$ 
		\epsilon^*(\bm{\tau}, \bm{Z}) = 
		\min_{n_0} [n_0/n: b^*(n_0 ; \bm{\tau}, \bm{Z}) \rightarrow \infty ]  
		$,
		represents the maximum proportion of observations that, when arbitrarily replaced, still provide bounded estimates \citep{donoho1983notion}.
        Our second result shows that our MIP approach for SFSOD achieves arbitrarily large BdP.
        \begin{proposition}[Breakdown point] \label{th:bdp}
            For any $\lambda$, $n$, $p$, $k_n$ and $k_p$,
            the the BdP of the $\hat{\bm{\beta}}$ estimator produced solving \eqref{eq:miqp} is
            $\epsilon^* = ( k_n + 1) / n$. 
        \end{proposition}
        \noindent
        Thus, $ k_n \geq n_0 $ is the only requirement to achieve the largest possible BdP.
        Similar results were obtained for the low-dimensional Least Quantile Estimator \citep{bertsimas2014least}, the LTS estimator with a Lasso penalty \citep{alfons2013sparse},
        and MM-estimators with a ridge or elastic net penalty 
        (\citealt{maronna2011robust}; \citealt{kurnaz2017robust}). 
        However, there are two caveats: the BdP can be misleading for non-equivariant estimators \citep{smucler2017robust}, and it only guarantees against the worst-case scenario -- infinite maximum bias -- as it does not account for the presence of large but finite biases in $\hat{\bm{\beta}}$. 
        This motivates the development of the robustly strong oracle property as well as optimal parameter estimation under the $L_2$-norm to provide additional theoretical guarantees. 
        \\
	   
        Next, we exclude the ridge-like penalty and take $\rho(\cdot) = \norm{ \cdot }_2^2$, making \eqref{eq:miqp} a \textit{Mixed-Integer Quadratic Program} (MIQP).  In this setting, we prove that under certain  conditions our approach satisfies the \textit{robustly strong oracle property} (see Definition~\ref{def:strongoracle}, based on \citealt{fan2014strong}). 
        In the following, we use the $L_0$ sparsity assumption on $\bm{\beta}$ and $\bm{\phi}$ as in \citet{zhang2012general}.
        Recall that MSOM leads to outlier removal (see Section~\ref{sec1:background}), and we showed in Proposition~\ref{lemma:1} that the $L_0$ constraint on $\bm{\phi}$ controls the level of trimming from the fit, thus this sparsity assumption on $\bm{\phi}$ is equivalent to the presence of MSOM outliers.
        In our SFSOD problem, as customary in feature selection literature, let $ \bm{\theta}_0 = ( \bm{\beta}^T_0 , \bm{\phi}^T_0 )^T \in \mathbb{R}^{p+n} $ be the true parameter vector, and decompose it
        as
        $
        \bm{\theta}_0 =  
        ( \bm{\theta}_{\mathcal{S}}^T , \bm{\theta}_{\mathcal{S}^c}^T )^T = 
        [ ( \bm{\beta}^T_{\mathcal{S}_\beta}  , \bm{\phi}^T_{\mathcal{S}_\phi} ) , ( \bm{\beta}^T_{\mathcal{S}^c_\beta} , \bm{\phi}^T_{\mathcal{S}^c_\phi} ) ]^T 
        $
        where
        $ \bm{\theta}_{\mathcal{S}} $ 
        contains only the true non-zero regression coefficients. 
        Let $\hat{\bm{\theta}}_0$ be the \textit{robust oracle estimator}
        $
        \hat{\bm{\theta}}_0 = (\bm{A}_{\mathcal{S}}^T \bm{A}_{\mathcal{S}})^{-1} \bm{A}_{\mathcal{S}}^T \bm{y}
        $,
        where 
        $ 
        \bm{A}_{\mathcal{S}} =  [\bm{X}_{\mathcal{S}_\beta}, \bm{I}_{\mathcal{S}_\phi} ]$ is 
        the $n \times ( p_0 + n_0 )$ matrix restricted to the active features belonging to
        $ \mathcal{S}_\beta$ and the outlying cases belonging to $\mathcal{S}_\phi 
        $.
        $\hat{\bm{\theta}}_0$ behaves as if the sets of active features and outliers are both known in advance. 
        
        \begin{definition}[Robustly strong oracle property]\label{def:strongoracle}
        An estimator $\hat{\bm{\beta}}$ 
        with support $\hat{\mathcal{S}}_\beta$ satisfies the robustly strong oracle property if (asymptotically) there exists tuning parameters which guarantee 
        $
        P ( \hat{\mathcal{S}}_\beta = \mathcal{S}_\beta ) \geq 
        P ( \hat{\bm{\beta}} = \hat{\bm{\beta}}_0 ) \to 1 
        $
        in presence of MSOM outliers.
        \end{definition}
        \noindent
        Such robust version of the oracle property is stronger and more general than the oracle property in the sense of \citet{fan2001variable},
        as it implies both SFSOD consistency and sign consistency (see also \citealt{bradic2011penalized}).
        Thus, SFSOD consistency depends on the achievability of the robust oracle estimator which we investigate by extending the theoretical framework developed by \citet{shen2013constrained} for feature selection.
        This requires weaker assumptions compared to other penalization methods \citep{zhang2012general}, and we generalize it to the presence of MSOM outliers. 
        Intuitively, if the robust oracle estimator is achievable (i.e.~if it has the lowest in-sample mean squared error for models of the same size), it is also the solution of our MIQP when the integer constraints are set to $k_p = p_0 $ and $k_n = n_0 $.
        Achievability depends on the difficulty of the problem, as measured by the \textit{minimal degree of separation} between the true model and a least favorable model -- indexed by the supports $ \mathcal{S} $ and $ \tilde{\mathcal{S}} $, respectively.
        This is defined as 
        \[
            \Delta_{\theta}(\bm{A}) = 
             \underset{\substack{\bm{\theta}_{\tilde{\mathcal{S}}}: \tilde{\mathcal{S}} \neq \mathcal{S}  ,
             \lvert \tilde{\mathcal{S}}_{\beta} \rvert \leq p_0 ,
             \lvert \tilde{\mathcal{S}}_{\phi} \rvert \leq n_0}}
             {\min}
            \frac{\norm{ \bm{A}_{\mathcal{S}} \bm{\theta}_{\mathcal{S}} - \bm{A}_{\tilde{\mathcal{S}}} \bm{\theta}_{\tilde{\mathcal{S}}} }^2_2}
            {n \max(\lvert \mathcal{S} \backslash \tilde{\mathcal{S}} \rvert, 1 )} ,
        \]
        which relates to the signal-to-noise-ratio and can be bounded as $ \Delta_{\theta} \leq \Delta_{\beta}  +  \Delta_{\phi} 
        $ (with $ \Delta_{\beta}$ and $ \Delta_{\phi}  $ defined similarly to $ \Delta_{\theta} $ using $\bm{X}$ and $\bm{I}_n$, respectively). We control this level of difficulty in Proposition~\ref{th:necessCond}, which provides a \textit{necessary} condition for 
        SFSOD
        consistency over $B(u,l) = \{ \bm{\theta} : \norm{ \bm{\theta} }_0 \leq u , \Delta_{\theta} \geq l \}$, the $L_0$-band with upper and lower radii $u$ and $l$, respectively (a subset of the $L_0$-ball $B(u,0)$ excluding a neighborhood of the origin). 
        \begin{proposition}[Necessary condition for SFSOD consistency] \label{th:necessCond}
            For any support estimate $ \hat{\mathcal{S}}$ and $u>l>0$, 
            $ \sup_{\bm{\theta}_0 \in B(u,l)}
                    P ( \hat{\mathcal{S}} = \mathcal{S} ) \to 1 $
            implies that
            \begin{equation} \label{eq:regul_cond}
                \Delta_{\theta} \geq l = \frac{\sigma^2}{n}  \max \left\{ d_\beta \log (p) , 
                 d_\phi \log (n) \right\}, 
            \end{equation}
            where $d_\beta > 0$  (which may depend on $\bm{X}$) and $d_\phi > 0$ are constants independent of $n$ and $p$.
        \end{proposition}
        \noindent
        This lower bound on $\Delta_{\theta}$ indicates one can focus on solving the most difficult task
        between outlier detection and feature selection; if this is achievable, \textit{a fortiori}, the other will be as well.
        Next, we provide a \textit{sufficient} condition for 
        SFSOD
        consistency based on a finite-sample result bounding the probability that our proposal differs from the robust oracle estimator.
        Note that for our formulation 
        $ 
        \mathcal{ \hat{S}}^{L_0} = \mathcal{S} \iff   \hat{\bm{\theta}}_{L_0} = \hat{\bm{\theta}}_0 
        $.
        \begin{proposition}[Oracle reconstruction for MIQP] \label{th:oracle_bound}
            For any $ n, p, n_0$ and $p_0$, the $\hat{\bm{\theta}}_{L_0}$ estimator produced solving \eqref{eq:miqp} with $k_p = p_0$ and $k_n = n_0$ is such that
            \begin{align} \label{eq:oracle_finite}
                    P\left( \hat{\bm{\theta}}_{L_0} \neq \hat{\bm{\theta}}_0 \right) 
                    \leq& \frac{  5e - 1 }{e-1}
                    \max \left\{ 
                        \exp{ \left( - \frac{n}{18 \sigma^2} \left[ \Delta_{\beta_0} - 36 \sigma^2 \frac{\log(p)}{n} \right] \right) } , \right. \nonumber \\ 
                        &   \left. \exp{ \left( - \frac{n}{18 \sigma^2} \left[ \Delta_{\phi_0} - 36 \sigma^2 \frac{\log(n)}{n} \right] \right) }
                     \right\}
                    .
            \end{align}
        \end{proposition}
        \noindent 
        Based on these results, one can easily prove the robustly strong oracle property as follows. 
         \begin{proposition}[MIQP robustly strong oracle property] \label{th:oracle_consistency}
                Assume that $ u _\theta= u_{\phi} + u_{\beta} $,
                where $u_{\phi} < n - k_p$ and 
                $u_{\beta} < \min(n - k_n, p)$, 
                and that there exists a constant $d_\theta > 36$ 
                such that $l_{\theta} = d_\theta \sigma^2 / n \max  \left\{ \log(p), \log(n) \right\}$.
                Then, under \eqref{eq:regul_cond} and for $(n,p) \to \infty$, the estimator $\hat{\bm{\theta}}_{L_0}$ produced solving \eqref{eq:miqp} with $k_p = p_0$ and $k_n = n_0$ satisfies
            \begin{enumerate}
                \item Robustly strong oracle property:
                    $
                    \sup_{\bm{\theta}_0 \in B(u_\theta,l_\theta)}
                    P ( \hat{\mathcal{S}}^{L_0} = \mathcal{S} ) \geq
                    \sup_{\bm{\theta}_0 \in B(u_\theta,l_\theta)}
                    P ( \hat{\bm{\theta}}_{L_0} = \hat{\bm{\theta}}_0 ) 
                    \to 1 $
                     uniformly over 
                    $
                      B(u_\theta,l_\theta) = \{ \bm{\theta} : \norm{\bm{\theta}}_0 = (p_0 + n_0) \leq u_\theta , \Delta_{\theta} \geq l_\theta \}  .
                    $

             \item Asymptotic normality:
                   $ 
                   \sqrt{n} ( \hat{\bm{\theta}}_{L_0} - \bm{\theta}_0 ) \to^d N ( \bm{0}, \bm{\Sigma}_{\theta}^{-1} )   
                   $,
                   where
                   $ \bm{\Sigma}_{\theta} = \sigma^2 
                    ( \bm{A}_{\mathcal{S}}^T \bm{A}_{\mathcal{S}} /n )^{-1}
                    $.
            \end{enumerate}
        \end{proposition}
        \noindent
        Proposition~\ref{th:oracle_consistency}(1) provides a sufficient condition for SFSOD consistency 
        and the robust oracle reconstruction
        up to a constant term $d_\theta $.
        Note that the number of features is allowed to exponentially increase with the sample size -- so these properties hold also in ultra-high dimensional settings where $ p = O(e^{n \alpha})$ with $\alpha = \frac{\Delta_{\theta}}{d_\theta \sigma^2} >0$.
        Proposition~\ref{th:oracle_consistency}(2) guarantees asymptotic efficiency of MIQP estimates, which achieve the Cramér--Rao lower bound as if the true sets of features and outliers are known a priori, and allows one to perform statistical inference.
        Importantly, existing penalized M-estimators provide weaker results under stronger assumptions \citep{,loh2017statistical,smucler2017robust,amato2020penalised}.
        We conclude with a result showing that our proposal attains optimal parameter estimation with respect to the $L_2$-norm in the presence of MSOM outliers.
        \begin{proposition}[Optimal parameter estimation] \label{th:proposition}
            Under the same conditions of Proposition~\ref{th:oracle_consistency}, the estimator $\hat{\bm{\theta}}_{L_0}$ produced solving \eqref{eq:miqp} with $k_p = p_0$ and $k_n = n_0$
            provides 
            \begin{enumerate}
                \item Optimal $L_2$-norm prediction error:
                 $
                    n^{-1} E \norm{ A(\hat{\bm{\theta}}_{L_0} - \bm{\theta}_0 ) }_2^2 
                    = \sigma^2 (p_0+n_0) / n
                  $.
                  \item Risk-minimax optimality for parameter estimation:
                  $
                    \sup_{\bm{\theta}_0 \in B(u_\theta,l_\theta)}
                        n^{-1} E \norm{ \bm{A}(\hat{\bm{\theta}}_{L_0} - \bm{\theta}_0 ) }_2^2 
                        = \inf_{\bm{\tau}_n }
                           \sup_{\bm{\theta}_0 \in B(u_\theta,l_\theta)} 
                            n^{-1} E \norm{ \bm{A}(\bm{\tau}_n - \bm{\theta}_0 ) }_2^2 = \sigma^2 u_\theta / n
                           .
                   $
              \end{enumerate}
        \end{proposition}
        \noindent
        Finally, we note that the theoretical guarantees developed in this Section can be extended in a similar fashion to other penalization methods,  albeit under stronger assumptions.
        For instance, one might consider the regularized $L_0$-penalty or the trimmed $L_1$-penalty.
        Importantly, our results do hold also when $p_0$ depends on $p$ and/or $n_0$ depends on $n$ which has yet to be established for other methods in the literature \citep{shen2013constrained}.
        We stress the fact that all results for the proposed formulation rely on the identification of the true $k_p$ and $k_n$ tuning parameters. While this is a standard requirement to establish oracle properties (see \citealt{fan2001variable}), it highlights the importance of proper tuning for these bounds (see Section~\ref{sec2:proposal}).

    \section{Simulation study} \label{sec3:sim}
	    
	    We use simulations to investigate the performance of our proposal and compare it with state-of-the-art heuristic methods.
    	The simulated data is generated as follows. 
    	The first column of the $n \times p$ design matrix $ \bm{X}$ comprises all $1$'s (for the model intercept) and we draw the remaining entries of each row independently from a $(p-1)$-variate Gaussian $N(\bm{0}, \bm{\Sigma}_X)$, we fix the values of the $p$-dimensional coefficient vector $\bm{\beta}$ as to comprise $p_0$ non-zero entries (including the intercept),
    	and we draw each entry of the $n$-dimensional error vector $\bm{\varepsilon}$ independently from  a univariate Gaussian $ N(\bm{0},\sigma^2_{\text{SNR}})$. 
    	Here the positive scalar $\sigma^2_{\text{SNR}}$ is used to modulate the signal-to-noise-ratio 
    	$  \text{SNR}  = \text{var}(\bm{X} \bm{\beta} ) / \sigma^2_{\text{SNR}}$ characterizing each experiment.
        Next, without loss of generality, we contaminate the first $n_0$ cases with an MSOM, adding the scalar mean shifts $\lambda_{\bm{\varepsilon}}$ and $\lambda_{\bm{X}}$, respectively, to the errors and each of the $p_0 - 1$ active predictors.

        Specific simulation scenarios are defined through the values of the parameters listed above. Here, we present results for
        $\bm{\Sigma}_X = \bm{I}_{p-1}$ (uncorrelated features),
        $p_0=5$ active features with $\beta_j = 2$ (without loss of generality these correspond to $j=1, \ldots, 5$), $\text{SNR}=5$, fraction of contamination $n_0/n = 0.1$, mean shifts $ \lambda_{\bm{\varepsilon}} = -10 $ and $ \lambda_{\bm{X}} = 10 $, increasing sample sizes $n= 50, 100, 150$, and a ``low''- and a ``high''-dimensional setting with $p=50, 200$.
        Results for additional simulation scenarios are provided in the Supplementary Material. 
    	
           Replicating each scenario a certain number of (independent) times, say $q$, and creating (independent) test data, say $(\bm{y}^*, \bm{X}^*)$, from the same generating scheme but without contamination, we compare methods with a variety of criteria, namely:
           (i) out-of-sample prediction performance, measured by the \textit{root mean squared prediction error}
           $
                \text{RMSPE} = 
                \sqrt{n^{-1} \sum_{i=1}^{n}\left(y_{i}^{*}-\bm{x}_{i}^{*} \hat{\bm{\beta}}\right)^{2}}$; 
            (ii) estimation accuracy for $\bm{\beta}$, measured by the {\em average mean squared error} $\text{MSE}(\hat{\bm{\beta}}) =  p^{-1} \sum_{j=1}^{p} \text{MSE}(\hat{\beta}_j)$, where for each $ \hat{\beta}_j $ we form
            $
        	\text{MSE}(\hat{\beta}_j) =  q^{-1}\sum_{i=1}^{q} \\ (\hat{\beta}_{ji} - \beta_{j})^2 =
        	(\bar{\beta}_j - \beta_j)^2 + q^{-1}
        		 \sum_{i=1}^{q} (\hat{\beta}_{ji} - \bar{\beta}_j)^2$, decomposed in squared bias and variance
        	(here $ \bar{\beta}_j = q^{-1} \sum_{i=1}^{q} \hat{\beta}_{ji}$); 
        	(iii) feature selection accuracy, measured by the {\em false positive rate} 
            $
           \text{FPR}(\hat{\bm{\beta}})= \lvert \{j \in\{1, \ldots, p\}: \hat{\beta}_{j} \neq 0 \wedge \beta_{j}=0 \} \rvert /  \lvert \{j \in\{1, \ldots, p\}: \beta_{j}=0 \} \rvert 
            $
            and the {\em false negative rate}
            $
           \text{FNR}(\hat{\bm{\beta}})=  \lvert \{j \in\{1, \ldots, p\}: \hat{\beta}_{j}=0 \wedge \beta_{j} \neq 0 \} \rvert / \lvert \{j \in\{1, \ldots, p\}: \beta_{j} \neq 0 \} \rvert
            $;
            (iv) outlier detection  accuracy,  which is similarly measured by $\text{FPR}(\hat{\bm{\phi}})$ and $\text{FNR}(\hat{\bm{\phi}})$;
            (v) computational burden, measured as CPU time in seconds (this is used as a rough evaluation, since software implementations of different methods are not entirely comparable). 

	        Using the robust oracle estimator as a benchmark,
	        we compare the following estimators:
	        (a) sparse-LTS \citep{alfons2013sparse}, 
	        (b) enet-LTS \citep{kurnaz2017robust}, 
	        and (c) our proposal described in Section~\ref{sec2:proposal} (denoted MIP in result tables).
	        All methods trim the true number of outliers ($k_n=n_0$) and only their feature sparsity level $k_p$ is tuned.
            Additional implementation details can be found in the Supplementary Material.
	            
            Table~\ref{tab:1means} provides means and standard deviations (SD) of simulation results over $q=100$ replications.
            Our proposal substantially outperforms competing methods in most criteria. 
            In particular, for the low-dimensional setting ($p=50$), 
            its RMSPE converges faster to the oracle solution and the variance of its $\hat{\bm{\beta}}$ decreases faster as $n$ increases (the bias is essentially non-existent for all methods). 
            Notably, the FPR($\hat{\bm{\beta}}$) of sparse-LTS and enet-LTS increases with the sample size, while our approach avoids these type II errors. 
            Even with these sparser solutions, we retain comparable (and at times lower) FNR($\hat{\bm{\beta}}$).
            Our method struggles most when $n=50$, suggesting that additional work for tuning MIP may be beneficial in under-sampled problems. 
            All methods perform very well in terms of FPR($\hat{\bm{\phi}}$) and FNR($\hat{\bm{\phi}}$), though enet-LTS is slightly worse for $n=50$.
            As expected, the computational burden of our procedure is substantially higher than that of the competing heuristic methods -- though we note that averages here are not representative, as there is a marked right skew in the distribution of computing times across replications.
            For comparison we provide medians and MAD in Table~B.3 of the Supplementary Material and find that results are even stronger.
            For example, the average computing time with $n=150$ and $p=200$ is 946.07 minutes compared to a median of 599 minutes. 
            Our experience suggests that the growth in computational burden is mainly due to increases in the absolute number of outliers as the sample size increases.
            
            Similar conclusions hold under the high-dimensional scenario with $n < p=200$. Our proposal outperforms other methods in most criteria and converges faster to the oracle solution as $n$ increases.
            The increased problem  complexity causes right skews in the distributions across replications (see again Table~B.3 of the Supplementary Material for results in terms  of medians and MADs), and increases the computational burden of all methods.
            Also notably, FPR($\hat{\bm{\phi}}$) and FNR($\hat{\bm{\phi}}$) for outlier detection are higher compared to the low-dimensional setting. 
            
            In the Supplementary Material we report results for additional simulation scenarios, e.g.,~with weaker SNR and collinear features. Our method outperformed others in most of these settings as well (see Supplementary Material for details and further discussion).
            \begin{table}[!ht]
            \centering
            \caption{Mean (SD in parenthesis) of RMSPE, variance and squared bias for $\hat{\beta}$, FPR and FNR for feature selection and outlier detection, and computing time, 
            based on 100 simulation replications.} 
            \label{tab:1means}
            \begingroup\scriptsize\renewcommand\arraystretch{0.5} 
            \setlength{\tabcolsep}{4pt}
            \begin{tabular}{cclcccccccc}
            \midrule \midrule
            \multicolumn{1}{c}{$n$} & \multicolumn{1}{c}{$p$} & \multicolumn{1}{c}{Method} & \multicolumn{1}{c}{RMSPE} &
            \multicolumn{1}{c}{$\text{var}(\hat{\bm{\beta}})$} &
            \multicolumn{1}{c}{$\text{bias}(\hat{\bm{\beta}})^2$} & 
            \multicolumn{1}{c}{$\text{FPR}(\hat{\bm{\beta}})$} & 
            \multicolumn{1}{c}{$\text{FNR}(\hat{\bm{\beta}})$} & 
            \multicolumn{1}{c}{$\text{FPR}(\hat{\bm{\phi}})$} & 
            \multicolumn{1}{c}{$\text{FNR}(\hat{\bm{\phi}})$} & 
            \multicolumn{1}{c}{Time} \\
            \midrule
            50 & 50 & Oracle & 1.91(0.29) & 0.01(0.00) & 0.00(0.00) & 0.00(0.00) & 0.00(0.00) & 0.00(0.00) & 0.00(0.00) &   0.00(0.00) \\ 
            &  & EnetLTS & 2.56(0.97) & 0.05(0.00) & 0.03(0.00) & 0.20(0.23) & 0.03(0.12) & 0.01(0.02) & 0.04(0.17) &  14.72(0.68) \\ 
            &  & SparseLTS & 2.52(0.44) & 0.06(0.00) & 0.00(0.00) & 0.53(0.07) & 0.00(0.00) & 0.00(0.01) & 0.00(0.00) &   1.53(0.21) \\ 
            &  & MIP & 2.22(0.80) & 0.04(0.00) & 0.00(0.00) & 0.00(0.00) & 0.08(0.16) & 0.00(0.00) & 0.00(0.00) &  11.21(9.86) \\ 
            [2mm]
            100 & 50 & Oracle & 1.86(0.16) & 0.00(0.00) & 0.00(0.00) & 0.00(0.00) & 0.00(0.00) & 0.00(0.00) & 0.00(0.00) &   0.00(0.00) \\ 
            &  & EnetLTS & 2.01(0.21) & 0.01(0.00) & 0.00(0.00) & 0.27(0.22) & 0.00(0.00) & 0.00(0.01) & 0.00(0.00) &  12.26(0.22) \\ 
            &  & SparseLTS & 2.15(0.20) & 0.03(0.00) & 0.00(0.00) & 0.66(0.08) & 0.00(0.00) & 0.00(0.01) & 0.00(0.00) &   2.23(0.34) \\ 
            &  & MIP & 1.91(0.32) & 0.01(0.00) & 0.00(0.00) & 0.00(0.00) & 0.01(0.06) & 0.00(0.00) & 0.00(0.00) &  40.62(30.78) \\ 
            [2mm]
            150 & 50 & Oracle & 1.81(0.16) & 0.00(0.00) & 0.00(0.00) & 0.00(0.00) & 0.00(0.00) & 0.00(0.00) & 0.00(0.00) &   0.00(0.00) \\ 
            &  & EnetLTS & 1.92(0.19) & 0.01(0.00) & 0.00(0.00) & 0.39(0.27) & 0.00(0.00) & 0.00(0.00) & 0.00(0.00) &  12.67(0.22) \\ 
            &  & SparseLTS & 1.99(0.19) & 0.02(0.00) & 0.00(0.00) & 0.68(0.08) & 0.00(0.00) & 0.00(0.00) & 0.00(0.00) &   4.32(0.84) \\ 
            &  & MIP & 1.82(0.20) & 0.00(0.00) & 0.00(0.00) & 0.00(0.00) & 0.00(0.03) & 0.00(0.00) & 0.00(0.00) & 433.42(246.43) \\ 
            [2mm]
            50 & 200 & Oracle & 1.89(0.25) & 0.00(0.00) & 0.00(0.00) & 0.00(0.00) & 0.00(0.00) & 0.00(0.00) & 0.00(0.00) &   0.00(0.00) \\ 
            &  & EnetLTS & 3.28(1.19) & 0.03(0.00) & 0.02(0.00) & 0.17(0.10) & 0.19(0.32) & 0.02(0.03) & 0.16(0.30) &  44.50(2.84) \\ 
            &  & SparseLTS & 2.95(1.04) & 0.02(0.00) & 0.01(0.00) & 0.16(0.01) & 0.08(0.21) & 0.01(0.03) & 0.09(0.25) &   3.44(0.54) \\ 
            &  & MIP & 2.44(1.30) & 0.02(0.00) & 0.00(0.00) & 0.00(0.01) & 0.11(0.25) & 0.01(0.03) & 0.07(0.23) &  23.32(33.84) \\ 
            [2mm]
            100 & 200 & Oracle & 1.81(0.20) & 0.00(0.00) & 0.00(0.00) & 0.00(0.00) & 0.00(0.00) & 0.00(0.00) & 0.00(0.00) &   0.00(0.00) \\ 
            &  & EnetLTS & 2.74(1.20) & 0.02(0.00) & 0.01(0.00) & 0.25(0.16) & 0.09(0.21) & 0.02(0.03) & 0.14(0.31) &  55.36(5.37) \\ 
            &  & SparseLTS & 2.31(0.26) & 0.01(0.00) & 0.00(0.00) & 0.31(0.02) & 0.00(0.00) & 0.00(0.01) & 0.00(0.00) &  11.29(2.25) \\ 
            &  & MIP & 1.88(0.36) & 0.00(0.00) & 0.00(0.00) & 0.00(0.00) & 0.02(0.07) & 0.00(0.00) & 0.00(0.00) & 421.32(826.02) \\ 
            [2mm]
            150 & 200 & Oracle & 1.81(0.15) & 0.00(0.00) & 0.00(0.00) & 0.00(0.00) & 0.00(0.00) & 0.00(0.00) & 0.00(0.00) &   0.00(0.00) \\ 
            &  & EnetLTS & 2.59(1.29) & 0.02(0.00) & 0.01(0.00) & 0.25(0.14) & 0.07(0.16) & 0.02(0.03) & 0.12(0.28) &  75.70(180.69) \\ 
            &  & SparseLTS & 2.25(0.20) & 0.01(0.00) & 0.00(0.00) & 0.41(0.04) & 0.00(0.00) & 0.00(0.00) & 0.00(0.00) &  18.91(2.74) \\ 
            &  & MIP & 1.84(0.23) & 0.00(0.00) & 0.00(0.00) & 0.00(0.00) & 0.01(0.03) & 0.00(0.00) & 0.00(0.00) & 946.07(987.32) \\ 
            \bottomrule 
            \end{tabular}
            \endgroup
            \end{table}
            
	\section{Application}\label{sec4:appl}
	    
	    In this Section we investigate the relation between childhood obesity and the human microbiome using our approach.
	    We use data from \citet{craig2018child} -- who linked child weight gain to oral and gut microbiota compositions as part of the \textit{Intervention Nurses Start Infants Growing on Healthy Trajectories} (INSIGHT) study \citep{paul2014intervention}. We note that all data are publicly available; raw microbiota reads were deposited in SRA under BioProject number PRJNA420339 (see \url{www.ncbi.nlm.nih.gov/bioproject/PRJNA420339}) and phenotype information was deposited under dbGaP study number phs001498.v1.p1 (see \url{www.ncbi.nlm.nih.gov/projects/gap/cgi-bin/study.cgi?study\_id=phs001498.v1.p1}). 
	    While previous work \citep{goodson2009obesity, haffajee2009relation, zeigler2012microbiota} focused on the relationship between adult and/or adolescent obesity and microbiome compositions, \citet{craig2018child} connected infant weight gain (which is known to be predictive of obesity later in life, \citealt{taveras2009weight}) to oral and gut microbiota of the child, as well as oral microbiome of the mother. 
	    The goal of our analysis is to further study these relationships -- selecting relevant bacterial types (features) while accounting for potential outliers in the data. 
	    
	    Based on the pre-processing  in \citet{craig2018child}, we retained 215 child oral, 189 child gut, and 215 maternal oral samples. Correspondingly, we considered the abundances of 77, 74 and 78 bacterial ``types'', respectively -- which the original authors obtained aggregating phylogenetically very sparse and correlated abundance data (we further filtered based on those with a MAD of zero).   
	    We also focused on one among the phenotypes studied in \citet{craig2018child}; namely, \textit{Conditional Weight Gain Score} (CWG) -- a continuous measure computed from weight gain between birth and six months, which is commonly used in pediatric research \citep{griffiths2009effects,savage2016effect} (a positive CWG indicates an accelerated weight gain). 
	    
	    We thus applied our approach along with enet-LTS, sparse-LTS and classical Lasso to three main models; the regressions of CWG on bacterial types abundances in child oral, child gut, and maternal oral samples.
	    The problem sizes were thus  $215\times 78, 189 \times 75$, and $215 \times 79$ with the addition of an intercept term. For each regression we split the data at random into a training ($n^{\text{tr}} \approx 0.8 n$) and a test ($n^{\text{te}} \approx 0.2 n$) set, and followed the approach in \citet{mishra2019robust} to produce a (robustly) \textit{scaled prediction error} estimate (SPE; see Supplementary Material for more details). 
	    We fixed  the trimming proportion to 10\% for robust procedures, and we repeated the analysis on five different training and test splits for all methods.
	    Table~\ref{tab:2} provides, for each of the three regressions, means and SD of results over the five splits -- including SPE, number of features selected on the training set ($\hat{p}_0^{\text{tr}}$), and number of {\em non-outlying} cases identified in the test set ($n^{\text{te}}-\hat{n}_{0}^{\text{te}} $);
	    medians and MADs are in Table~C.1 of the Supplementary Material. 
	    The last column contains the total number of features selected on the full data ($ \hat{p}_0^{\text{full}} $).
	    \begin{table}[!ht]
            \centering
            \caption{
            Mean (SD in parenthesis) of SPE, number of features selected on the training set, and number of non-outlying points identified in the test set, based on five train-test splits. 
            Last column: number of features selected on the full data. 
            Robust methods use 10\% trimming.}
            \label{tab:2}
            \begingroup\scriptsize\renewcommand\arraystretch{0.5} 
            \setlength{\tabcolsep}{12.2pt}
              \begin{tabular}{lccclcccc}
             \midrule \midrule
             \multicolumn{1}{c}{Data} & \multicolumn{1}{c}{$n^{\text{tr}}$} & \multicolumn{1}{c}{$n^{\text{te}}$} & \multicolumn{1}{c}{$p$} &
             \multicolumn{1}{c}{Method} &
             \multicolumn{1}{c}{SPE} & 
             \multicolumn{1}{c}{$\hat{p}_0^{\text{tr}}$} & 
             \multicolumn{1}{c}{$n^{\text{te}}-\hat{n}_{0}^{\text{te}} $}  &
             \multicolumn{1}{c}{$\hat{p}_0^{\text{full}}$}\\
              \midrule
             Child oral & 172 & 43 & 78 & EnetLTS & 0.27(0.11) & 58.8(0.86) & 34.0(1.79) & 60\\ 
             & &  &  & SparseLTS & 0.34(0.08) & 57.4(6.2) & 33.8(1.50) & 55\\ 
             & & &  & MIP & 0.30(0.04) & 2.8(0.8) & 34.4(1.33) & 2\\ 
             & &  &  & Lasso & 0.52(0.05) & 3.6(2.14) & 34.4(1.60) & 10\\
               [2mm]
             Child gut & 152 & 38 & 75 & EnetLTS & 0.43(0.13) & 57.6(0.24) & 28.6(1.17) & 54\\ 
             & &  &  & SparseLTS & 0.33(0.02) & 57.4(4.69) & 30.0(1.30) & 74\\ 
             & & &        & MIP & 0.29(0.04) & 4.0(1.26) & 30.8(0.97) & 2\\ 
             & &  &  & Lasso & 0.60(0.13) & 3.2(1.05) & 31.4(1.54) & 1\\
               [2mm]
             Maternal oral &  172 & 43 & 79 & EnetLTS & 0.42(0.03) & 63.6(1.63) & 34.2(1.07) & 65\\ 
             & &  &  & SparseLTS & 0.42(0.14) & 56.6(5.32) & 34.2(1.02) & 60\\ 
             & & &  & MIP & 0.35(0.09) & 4.0(1.05) & 36.8(0.80) & 6\\ 
             & &  &  & Lasso & 0.75(0.09) & 1.0(0.00) & 37.2(0.97) & 1\\
               \bottomrule 
            \end{tabular}
            \endgroup
        \end{table}
        
        Our proposal outperforms competing methods for the child gut and maternal oral regressions (with SPEs lower by 14-52\% and 15-53\%, respectively). For the child oral regression, it has SPE 13\% higher than enet-LTS but still produces lower SPE than sparse-LTS and Lasso by 10\% and 42\% -- while selecting substantially sparser solutions (we also note that based on the median SPE in Table~C.1 our method is actually 22\% better than enet-LTS). Unsurprisingly, the non-robust Lasso results in the weakest solutions in terms of prediction error. 
        Outliers within the training data of each split may pull it towards very sparse but poorly predictive solutions. In comparison, our solutions are equally sparse (we select on average 1-3 features in addition to the intercept) but with much better prediction power. 
        
        When applied to the full datasets, enet-LTS and sparse-LTS select a very large number of bacterial type abundances, hindering interpretation. In contrast, our procedure produces very sparse solutions. In the child oral and gut microbiota, it selects only one bacterial type -- in both cases belonging to the Firmicutes phylum and having a positive effect on CWG (\citealt{craig2018child} identified different bacteria types, but always from the Firmicutes and with positive effects).
        In the mother oral microbiome, our method selects five  bacterial types -- one from the Bacteroidetes (positive effect), two from the Firmicutes (both positive), one from the Fusobacteria (negative), and one from the Proteobacteria (negative). These too do not coincide with the bacterial types selected in \citet{craig2018child} -- except for the Fusobacteria type.
        However, other Bacteriodetes and Proteobacteria types were selected by \citet{craig2018child} and the signs of the effects are all consistent with ours.
        
        The Lasso also produces sparse solutions, though for the child gut and mother oral microbiota it only selects the intercept. In the child oral microbiome, Lasso selects 9 bacterial types -- one from Actinobacteria (positive effect), two from Bacteriodetes (both negative), four from Firmicutes (three negative, one positive), and two from Proteobacteria (one positive, one negative) groups. Of those, two of the Firmicutes types were also selected in \citet{craig2018child}, and one other was selected by our approach. The effects of these types do not appear consistent with those found in \citet{craig2018child} -- even at the broad level of phyla.  For instance, Bacteriodetes types were found to have positive effects in previous work (and in our results), but the effect is negative in the Lasso solution. Similarly, Firmicutes types were found to have positive effects, but in the Lasso solution three of the four selected ones have negative effects. Overall, these results demonstrate not only that our proposal is very competitive in terms of predictive power (compared to other robust and non-robust methods), but also that it provides parsimonious, interpretable and informative solutions consistent with findings in existing literature.

	\section{Final  remarks}\label{sec5:final}
	    
	    Our proposal provides a general framework to simultaneously perform sparse estimation and outlier detection that can be used for Linear Models, as well as Generalized Linear Models and several classification and non-parametric methods. In our main results, we focus specifically on Linear Models (as do existing heuristic approaches) -- but we directly tackle the original problem and preserve its discrete nature; this facilitates implementation, interpretation, and generalizations.
	    Importantly, we provide optimal guarantees from both optimization and theoretical perspectives, and verify that these hold in numerical experiments.
	    
	    Specifically, our approach relies on $L_0$ constraints -- extending prior work where they were used separately for feature selection our outlier detection.
	    Our simultaneous MIP formulation can handle problems of considerable size, and produces solutions that improve upon existing heuristic methods.
	    Although our formulation provides provably optimal solutions from the optimization perspective, it is crucial to tune its integer constraints. 
	    Thus, we also provide computationally efficient, data-driven approaches to induce sparsity in the coefficients and the estimated amount of contamination.
	    
	    Theoretical properties characterizing our proposal include its high breakdown point, the \textit{robustly strong oracle property} -- which holds in ultra-high dimensional settings where the number of predictors grows exponentially with the sample size -- and optimality in parameter estimation with respect to the $L_2$-norm (i.e.~optimal prediction error and risk-minimaxity).
	    Notably, our proposal requires weaker assumptions than prior methods in the literature and, unlike such methods, it allows the sparsity level and/or the amount of contamination to grow with the number of predictors and/or the sample size.
	    
	    In addition to  performing numerical experiments, we presented a real-world application investigating the relationship between childhood obesity and the human microbiome.
	    Our proposal outperformed existing heuristic methods in terms of predictive power, robustness and solution sparsity, and produced results consistent with prior childhood obesity studies.

	    The work presented here can be expanded in several directions.
	    Even with modern solvers, larger problems and optimal tuning can make utilizing MIPs computationally challenging. We are pursuing ways to reduce the computational burden -- e.g.,~efficiently and effectively exploring the graph built by Branch \& Bound algorithms \citep{gatu2007graph}, extending  the \textit{perspective formulation} \citep{frangioni2006perspective} to the presence of MSOM outliers, and generating high-quality initial solutions for warm-starts and big-$\mathcal{M}$ bounds through continuous methods \citep{bertsimas2014least}.
	    To improve solution quality, we are further exploring the addition of a ridge-like term, which would naturally benefit from the extension of the perspective formulation, as well as robust versions of whitening methods for feature de-correlation \citep{kenney2018efficient}.
	    Finally, we are particularly interested in the class of Generalized Linear Models and Gaussian graphical models,
	    where the use of $L_0$ constraints for sparse estimation has already been investigated from a theoretical perspective \citep{shen2012likelihood}.
	    However, an effective implementation in modern MIP solvers is not trivial and the possible presence of adversarial contamination has not received much attention in the literature. 
	
    
    \def\spacingset#1{\renewcommand{\baselinestretch}%
    {#1}\small\normalsize} \spacingset{1}

	\bigskip
    \begin{center}
    {\Large\bf SUPPLEMENTARY MATERIAL}
    \end{center}

\spacingset{1}
    
\renewcommand{\theHsection}{A\arabic{section}}
\appendix

    \phantomsection
	\section*{Appendix A: Theoretical results}
	\renewcommand{\theequation}{A.\arabic{equation}}
	\setcounter{equation}{0}

        \begin{proof}[Proof of Proposition~\ref{lemma:1}.]
	        For the considered loss function $\rho(\cdot)$, with an additional ridge penalty and $L_0$ constraints on $ \bm{\beta}$ and $ \bm{\phi}$,
	        any feasible solution behaves similarly to the unpenalized OLS case --
	        whose proof is usually based on the Sherman--Morrison formula \citep{atkinson1985plots,chatterjee1988sensitivity}. 
	        Indeed, the additional $k_n$ degrees of freedom in \eqref{eq:miqp} are used to zero-out the largest transformed residuals as in \eqref{eq:quival_lts}. We can write \eqref{eq:miqp} as follows
	        \begin{align}
	            \min_{\substack{\norm{\bm{\beta}}_0\leq k_p \\\norm{\bm{\beta}}_2^2\leq\lambda
	            \\\norm{\bm{\phi}}_0\leq k_n}} \frac{1}{n} \sum_{i=1}^n \rho(y_i-\bm{x}_i^T\bm{\beta}-\phi_i) &= 
	            \min_{\substack{||\bm{\beta}||_0\leq k_p \\||\bm{\beta}||_2^2\leq\lambda}}\min_{||\bm{\phi}||_0 \leq k_n}\frac{1}{n}\sum_{i=1}^n \rho(y_i-\bm{x}_i^T\bm{\beta}-\phi_i) \nonumber \\
	            &= \min_{\substack{||\bm{\beta}||_0\leq k_p \\||\bm{\beta}||_2^2\leq\lambda}}\min_{||\bm{\phi}||_0 \leq k_n}\frac{1}{n}\sum_{i=1}^n \rho(e_i-\phi_i) \nonumber \\
	            &= \min_{\substack{||\bm{\beta}||_0\leq k_p \\||\bm{\beta}||_2^2\leq\lambda}} T(\bm{\beta}) \nonumber .
	        \end{align}
	       Then for fixed and feasible $\bm{\beta}$, we evaluate $T(\bm{\beta})$ before performing the outer minimization. 
	       Since $\rho(\bm{x}) \geq 0$ and $\rho(\bm{x})=0$ at $\bm{x}= \bm{0}$, each $\rho(e_i-\phi_i)$ in $T(\bm{\beta})$ is minimized at $\phi_i = e_i$. 
	       However, only at most
	       $k_n$ of the $\rho(e_i-\phi_i)$ with $\rho(e_i) \neq 0$ can achieve this minimum as $n-k_n$ or more of $\phi_i$ must be set to 0. 
	       Hence, it follows that 
	       $
	           T(\bm{\beta})
	           = n^{-1} [\sum_{i=1}^{n-k_n} (\rho(e_i-0))_{i:n} + \sum_{i=n-k_n+1}^{n} (\rho(e_i-e_i))_{i:n} ] 
	           = n^{-1}\sum_{i=1}^{n-k_n} (\rho(e_i))_{i:n}
           $.
	    \end{proof}

	   \begin{proof}[Proof of Proposition~\ref{th:bdp}.]
            
            Our approach is similar to \citet{bertsimas2014least} and \citet{alfons2013sparse} and is not affected by the presence of an $L_0$ constraint on $\bm{\beta}$.
            For the first part of the proof, consider the presence of $ n_0 \leq k_n = n-h $ outliers, so that $n-n_0 \geq h $ units are non-outlying. 
            Let
            $ \tilde{\bm{Z}} = [\tilde{\bm{y}}, \tilde{\bm{X}} ]$
            be the contaminated sample and 
            $M_y = \max_{1, \ldots, n}{ \lvert y_i \lvert } $.
            If $\hat{\bm{\beta}} = \bm{0}$ (possibly excluding the intercept term), the corresponding loss
            $ Q(\hat{\bm{\beta}}) $ in \eqref{eq:miqp} satisfies
            $
                Q(\bm{0}) = 
                \sum_{i=1}^{h} [ \rho (\tilde{\bm{y}} ) ]_{1:n} 
                \leq \sum_{i=1}^{h} [\rho (\bm{y}) ]_{1:n} 
                \leq h \rho (M_y) ,
            $ 
            where the first inequality follows from the fact that the contamination is arbitrary but not necessarily adversarial, as well as the result in Proposition~\ref{lemma:1}.
            Take any other estimate $\hat{\bm{\beta}}$ such that 
            $ \norm{\hat{\bm{\beta}}}_2^2 \geq l_1 $, where 
            $ l_1 = ( h \rho (M_y) + 1 )  / \lambda^* $ is independent from the contamination structure.
            It follows that
            $
               Q( \hat{\bm{\beta}} ) \geq
               \lambda^* \norm{ \bm{\beta} }_2^2 \geq
                h \rho (M_y) + 1  > Q( \bm{0} ) ,
            $
            leading to a contradiction since 
            $ Q( \bm{0} ) \geq Q( \hat{ \bm{\beta} } ) $
            (i.e.~the objective is non-decreasing in the number of active features).
            Thus, $ \norm{ \hat{ \bm{\beta}}(\tilde{\bm{Z}}) }_2^2 \leq l_1 $ shows that the MIP estimator in \eqref{eq:miqp} does not breakdown for $n_0 \leq k_n$ (i.e.~$ \epsilon^* \geq ( k_n + 1) / n $).
            
            For the second part of the proof, we need to show that 
            $ \epsilon^* \leq ( k_n + 1) / n $.
            We assume that the estimator does not breakdown,
            so that
            $ \norm{ \hat{ \bm{\beta} }(\tilde{\bm{Z}}) }_2^2 \leq u_1 $.
            Let $ n_0 > k_n $ and denote the corresponding contaminated sample as
            $ \tilde{\bm{Z}} = (\tilde{\bm{y}}, \tilde{\bm{X}}) = \bm{Z} + (\bm{\Delta}_{y}, \bm{\Delta}_{X} ) $.
            It follows that
            \begin{align}
                Q( \hat{\bm{\beta}}) =&
                \sum_{i=1}^{n-n_0} \left[ \rho (\tilde{\bm{y}} - \tilde{\bm{X}} \hat{\bm{\beta}} ) \right]_{i=1:n} 
                + \sum_{j=n-n_0+1}^{n-k_n} \left[ \rho (\tilde{\bm{y}} - \tilde{\bm{X}} \hat{\bm{\beta}} ) \right]_{j=1:n}  
                + \lambda^* \norm{\hat{\bm{\beta}}}_2^2  \nonumber \\
                \geq& 
                \left[ \rho \left( (y_i - \bm{x}_i^T \hat{\bm{\beta}}) + (\Delta_{y_i} - \bm{\Delta}_{x_i}^T \hat{\bm{\beta}} ) \right) \right]_{i=n-n_0+1}
                + \lambda^* \norm{\hat{\bm{\beta}}}_2^2  \label{eq:proof_bdp}
            \end{align}
            since at least one outlier is included in the fit;
            namely, the unit corresponding to the $(n-n_0+1)$-th position of the ordered transformed residuals. 
            Thus, the possible unboundedness of \eqref{eq:proof_bdp}, as both terms $ \Delta_{y_i} $ and $ \bm{\Delta}_{x_i} $ may take arbitrarily large values, contradicts the assumption that $ \norm{ \hat{ \bm{\beta} }(\tilde{\bm{Z}}) }_2^2 \leq u_1 $.
	   \end{proof}

 	    \begin{proof}[Proof of Proposition~\ref{th:necessCond}.]
 	        
	        The following proof provides a necessary condition for any method to achieve SFSOD consistency.
            This has been proved for feature selection with an $L_0$ constraint  \citep{shen2012likelihood,shen2013constrained}, and also extended to 
            the presence of group constraints \citep{xiang2015efficient}.
            Here we extend this further to account for the presence of, and identify, MSOM outliers.
            The main difference being that variable selection in \eqref{eq:miqp} is performed on two disjoint sets of coefficients $ \beta$ and $\phi$.
            Similarly to Theorem~1 in \citet{shen2013constrained}, we consider a least favorable scenario where Fano's inequality can be applied.
            Let $ \mathcal{C}_\beta = \{ \bm{\beta}_j \}_{j=0}^{p} $ be a collection of parameters 
            with components equal to $ \gamma_\beta $ or $0$ (e.g.,~one can think of $\bm{\beta}_0$ as being the true model).
            Similarly, 
            define $ \mathcal{C}_\phi = \{ \bm{\phi}_i \}_{i=0}^{n} $ as a collection of parameters 
            with components equal to $ \gamma_\phi $ or $0$.
            Assume also that 
            $ \norm{ \bm{\beta}_j - \bm{\beta}_{j'} }^2_2 \leq 4 \gamma^2_\beta $ 
            for any $j, j' \in {0, 1, \ldots, p}$,
            and $ \norm{ \bm{\phi}_i - \bm{\phi}_{i'} }^2_2 \leq 4 \gamma^2_\phi $ 
            for any $i, i' \in {0, 1, \ldots, n}$.
            Let
            $ 
            \Delta_{\beta}^* =
            \min_{\bm{\beta}_0: \lvert \bm{\beta}_0 \rvert \geq 1, \bm{\beta}_0 \in \mathcal{S}_\beta, \lvert \mathcal{S}_\beta \rvert \leq p_0 } \Delta_{\beta}
            $
            and
            $
             \Delta_{\phi}^* = 
            \min_{\bm{\phi}_0: \lvert \bm{\phi}_0 \rvert \geq 1, \bm{\phi}_0 \in \mathcal{S}_\phi, \lvert \mathcal{S}_\phi \rvert \leq n_0 } \Delta_{\phi}
            $
            such that
            $ 
              r_{\beta} =  \max_{1 \leq j \leq p}  \norm{\bm{X}_j}_2^2  [n \Delta_{\beta}^*]^{-1}
            $
            and 
            $ 
                r_{\phi} = [ n \Delta_{\phi}^* ]^{-1} 
            $.
            For any $\bm{\beta}_j , \bm{\beta}_{j'} \in \mathcal{C}_\beta $ with densities $q(\bm{\beta}_j)$ and $q(\bm{\beta}_{j'})$, respectively, the corresponding Kullback-Leibler information is equal to
            $ D[q(\bm{\beta}_j), q(\bm{\beta}_{j'})] 
            = \norm{ \bm{X} (\bm{\beta}_j - \bm{\beta}_{j'} ) }^2_2
            \leq   2 \max_{1\leq k \leq p} \norm{\bm{X}_k}^2_2  \gamma^2_\beta  / (n \sigma^2)  
            \leq   2 r_\beta \Delta_{\beta}  / \sigma^2 $.
            Here the first bound is obtained using sub-additivity and the triangle inequality, and the second one is based on Lemma~1 in \citet{shen2013constrained}.
            Similarly, 
            $ D[q(\bm{\phi}_i), q(\bm{\phi}_{i'})] 
            = \norm{ \bm{\phi}_i - \bm{\phi}_{i'}  }^2_2
            \leq   2  \gamma^2_\phi  / ( n \sigma^2 )
            \leq  2 r_\phi \Delta_{\phi} / \sigma^2 $.
            Thus, for any estimates $\bm{T}_\beta$ and $\bm{T}_\phi$, it follows from Fano's inequality that: 
            $ F_\beta = (p+1)^{-1} \sum_{j \in \mathcal{C}_\beta} P (\bm{T}_\beta=j) \leq  
                (2 n r_\beta \Delta_{\beta}  + \sigma^2 \log 2 ) /
                ( \sigma^2 \log (p) ) $ 
            and 
            $ F_\phi = (n+1)^{-1} \sum_{i \in \mathcal{C}_\phi} P (\bm{T}_\phi=i)  \leq  
                ( 2 n r_\phi \Delta_{\phi}  + \sigma^2 \log 2 ) /
                ( \sigma^2 \log (n) ) $.
            Using the fact that
            $ 
                P ( \hat{\mathcal{S}} \neq \mathcal{S} )  = 
                P [ ( \hat{\mathcal{S}}_\beta \neq \mathcal{S}_\beta ) \cup  
                   ( \hat{\mathcal{S}}_\phi \neq \mathcal{S}_\phi ) ]  =
                1 - P [ ( \hat{\mathcal{S}}_\beta = \mathcal{S}_\beta )  \cap  
                   ( \hat{\mathcal{S}}_\phi = \mathcal{S}_\phi ) ] \geq 
                1 - \min \{ P(\hat{\mathcal{S}}_\beta = \mathcal{S}_\beta) ,  
                 P( \hat{\mathcal{S}}_\phi = \mathcal{S}_\phi) \}
             $ leads to the following lower bound:
            \begin{align} \label{eq:fano}
                \sup_{ \{ (\theta, A) : \Delta_{\theta} \leq R^* \} }   
                P ( \hat{\mathcal{S}} \neq \mathcal{S} )  
                \geq  1 - \min(F_\beta, F_\phi) , 
            \end{align}
            where 
            $R^* = \max \{ (1-c^*_\beta) \sigma^2 \log(p) / ( 2 n r_\beta ),
            (1-c^*_\phi) \sigma^2 \log(n) / ( 2 n r_\phi ) \} $
            and $c^*_\beta , c^*_\phi  > 0$.
            For $ \sup_{ \theta_0 \in B_0(u,l)} (\mathcal{\hat{S}} \neq \mathcal{S} ) \to 0 $ as in Proposition~\ref{th:necessCond}, it follows from \eqref{eq:fano} that
            the $L_0$-band $ B(u,l) $ cannot interact with the $L_0$-ball $ B(R^*, 0) $.
            Thus, 
            a necessary condition for any estimator to achieve SFSOD consistency is that
            $ l \geq  \sigma^2 / n \max \left\{ \log(p) / (4 r_\beta) , 
                \log(n) / 4 r_\phi  \right\} 
            $; this provides a tighter bound compared to a na\"ive substitution of $p+n$ in place of $p$ in Theorem~1 of \citet{shen2013constrained}.
        \end{proof}

 	    \begin{proof}[Proof of Proposition~\ref{th:oracle_bound}.]
 	    
            The following result bounds the reconstruction error 
            $ P( \hat{\bm{\theta}}_{L_0} \neq \hat{\bm{\theta}}_0 ) \geq
            P( \hat{\mathcal{S}}^{L_0} \neq \mathcal{S} ) $
            and extends Theorem~2 in \citet{shen2013constrained} to the presence of MSOM outliers.
            Let $ \bar{\mathcal{S}} \subset \{1,  \dots, (p+n)\}$ be any feasible estimate of the active set
            such that $ \bar{\mathcal{S}}_\beta \neq \mathcal{S}_\beta $
            and $ \bar{\mathcal{S}}_\phi \neq \mathcal{S}_\phi $, with
            $ \lvert \bar{\mathcal{S}}_\beta \rvert \leq p_0 $ and 
            $ \lvert \bar{\mathcal{S}}_\phi \rvert \leq n_0 $.
            Note that if $k_p = p_0$ and $k_n = n_0$, it follows that
            $ \lvert \mathcal{\hat{S}}^{L_0} \rvert = 
            \lvert \mathcal{\hat{S}}_\beta^{L_0} \rvert 
            + \lvert \mathcal{\hat{S}}_\phi^{L_0} \rvert \leq p_0 + n_0 $.
            To simplify the notation, take 
            $ L(\bm{\theta}, \mathcal{S}_\beta, \mathcal{S}_\phi ) = \frac{1}{2} \norm{ \bm{y} - \bm{X}_{\mathcal{S}_\beta} \bm{\beta}_{\mathcal{S}_\beta} -
            \bm{I}_{n_{\mathcal{S}_\phi }}
            \bm{\phi}_{\mathcal{S}_\phi }}^2_2 $.
            Partitioning $ \bar{\mathcal{S}} $ as
            $ ( \bar{\mathcal{S}}_\beta \backslash \mathcal{S}_\beta ) \cup  
            ( \bar{\mathcal{S}}_\beta \cap \mathcal{S}_\beta ) \cup 
            ( \bar{\mathcal{S}}_\phi \backslash \mathcal{S}_\phi ) \cup  
            ( \bar{\mathcal{S}}_\phi \cap \mathcal{S}_\phi ) $,
            it follows that:
            \begin{align}
                P \left( \hat{\bm{\theta}}_{L_0} \neq \hat{\bm{\theta}}_0 \right)  \leq&
                P \left( \hat{\bm{\beta}}_{L_0} \neq \hat{\bm{\beta}}_0 \right) + 
                    P \left( \hat{\bm{\phi}}_{L_0} \neq \hat{\bm{\phi}}_0  \right) \nonumber \\
                    \leq& \sum_{\hat{\mathcal{S}}_\beta^{L_0} \in \bar{\mathcal{S}}_\beta} P \left( L(\bm{\beta}, \hat{\mathcal{S}}_\beta^{L_0} ) - L(\bm{\beta}_0, \mathcal{S}_\beta ) \leq 0 \right) \nonumber \\
                 & + 
                 \sum_{\hat{\mathcal{S}}_\phi^{L_0} \in \bar{\mathcal{S}}_\phi} P \left( L(\bm{\phi}, \hat{\mathcal{S}}_\phi^{L_0} ) - L(\bm{\phi}_0, \mathcal{S}_\phi) \leq 0 \right) 
                 \nonumber \\
                 \leq& \sum_{i_p = 0}^{p_0 -1} \sum_{j_p = 0}^{p_0 - i_p} 
                 \binom{p - p_0}{j_p} \binom{p_0}{i_p}
                P \left( L(\bm{\beta}, \hat{\mathcal{S}}_\beta^{L_0} = \bar{\mathcal{S}}_\beta(i_p, j_p) )  - L(\bm{\beta}_0, \mathcal{S}_\beta) \leq 0   \right) \nonumber \\
                & + 
                \sum_{i_n = 0}^{n_0 -1} \sum_{j_n = 0}^{n_0 - i_n}
                \binom{n - n_0}{j_n} \binom{n_0}{i_n}
                P \left( L(\bm{\phi}, \hat{\mathcal{S}}_\phi^{L_0} = \bar{\mathcal{S}}_\phi(i_n, j_n) )  - L(\bm{\phi}_0, \mathcal{S}_\phi)  \leq 0 \right)
               \nonumber ,
            \end{align}
            where the first inequality follows from the union bound, the second inequality uses the probability of each feasible solution,
            and the third upper bound is based on the total number of possible solutions for a given size of correct ($i_p$ and $i_n$) and incorrect ($j_p$ and $j_n$) selections for SFSOD.
            Following the argument in \citet{shen2013constrained} separately on the two terms
            $ P ( \hat{\bm{\beta}}_{L_0} \neq \hat{\bm{\beta}}_0 ) $ and  
            $ P ( \hat{\bm{\phi}}_{L_0} \neq \hat{\bm{\phi}}_0  ) $
            leads to the following result:
            \begin{align}
                P \left( \hat{\bm{\theta}}_{L_0} \neq \hat{\bm{\theta}}_0 \right) \leq& 2 
                \sum_{i_p = 1}^{p_0} \sum_{j_p = 0}^{i_p} 
                (p - p_0)^{j_p} p_0^{i_p}
                \exp \left( - \frac{i_p}{18 \sigma^2} n \Delta_{\beta} + \frac{2}{3} j_p \right) \nonumber \\
                & + 
                2 \sum_{i_n = 1}^{n_0} \sum_{j_n = 0}^{i_n}
                (n - n_0)^{j_n}  n_0^{i_n} 
                \exp \left( - \frac{i_n}{18 \sigma^2} n \Delta_{\phi} + \frac{2}{3} j_n \right) \nonumber \\
                \leq & 
                \frac{2 e}{ e- 1 } 
                \left\{
                R\left( \exp \left[ - \frac{n}{18 \sigma^2} \left( \Delta_{\beta} - 36 \frac{\log p}{ n} \sigma^2 \right)  
                \right] \right) \right. \nonumber \\
                & \left. + R\left( \exp \left[ - \frac{n}{18 \sigma^2} \left( \Delta_{\phi} - 36 \frac{\log n}{ n} \sigma^2 \right)  \right] \right)  \right\}  \nonumber \\
                \leq & 
                \frac{4 e}{ e- 1 } 
                \max \left\{
                R\left( \exp \left[ - \frac{n}{18 \sigma^2} \left( \Delta_{\beta} - 36 \frac{\log p}{ n} \sigma^2 \right)  
                \right] \right) , \right. \nonumber \\
                & \left. R\left( \exp \left[ - \frac{n}{18 \sigma^2} \left( \Delta_{\phi} - 36 \frac{\log n}{ n} \sigma^2 \right)  \right] \right)  \right\}  \nonumber ,
            \end{align}
            where $R(x) = x / (1-x)$.
            Using the fact that $ P ( \hat{\bm{\theta}}_{L_0} \neq \hat{\bm{\theta}}_0 ) \leq 1$ establishes the result in \eqref{eq:oracle_finite}; 
            this is a tighter bound compared to the na\"ive extension of Theorem~2 in \citet{shen2013constrained} using $p+n$ in place of $p$.
	    \end{proof}

 	    \begin{proof}[Proof of Proposition~\ref{th:oracle_consistency}.]
	        The result in Proposition~\ref{th:oracle_consistency}(1) immediately follows from Proposition~\ref{th:oracle_bound} through a pointwise bound of \eqref{eq:oracle_finite} to $\bm{\theta}_0 \in B(u_\theta,l_\theta) $.
	        For Proposition~\ref{th:oracle_consistency}(2) our approach is similar 
	        to \citet{liu2013asymptotic} and \citet{zhu2020high}. 
	        Proposition~\ref{th:oracle_consistency}(1) guarantees that 
	        $ P( \hat{\bm{\theta}}_{L_0} = \hat{\bm{\theta}}_0 ) \to 1 $ as $(n,p) \to \infty$.
	        Therefore, with a probability tending to one, it follows that:
	        \begin{align}
	          \hat{\bm{\theta}}_{L_0} = &( \bm{A}_{\hat{\mathcal{S}}^{L_0}}^T \bm{A}_{\hat{\mathcal{S}}^{L_0}} )^{-1} 
	                                  \bm{A}_{\hat{\mathcal{S}}^{L_0}}  \bm{y} \nonumber \\
                                  &( \bm{A}_{\hat{\mathcal{S}}^{L_0}}^T \bm{A}_{\hat{\mathcal{S}}^{L_0}} )^{-1} 
	                                  \bm{A}_{\hat{\mathcal{S}}^{L_0}} ( \bm{A}_{\mathcal{S}} \bm{\theta}_0 + \bm{\varepsilon} ) \nonumber \\
	                             & \bm{\theta}_0 + ( \bm{A}_{\mathcal{S}}^T \bm{A}_{\mathcal{S}} )^{-1} 
	                                            \bm{A}_{\mathcal{S}} \bm{\varepsilon} \nonumber .
	        \end{align}
	        Using the moment generating function with the fact that 
	        $ \varepsilon_i \sim N(0, \sigma^2) $ for $ i \notin \mathcal{S}_\phi$
	        leads to 
            $ ( \bm{A}_{\mathcal{S}}^T \bm{A}_{\mathcal{S}} )^{-1} \bm{A}_{\mathcal{S}} \bm{\varepsilon} \sim   
	        N(0, \sigma^2 ( \bm{A}_{\mathcal{S}}^T \bm{A}_{\mathcal{S}} )^{-1} ) $.
	        Consequently, 
	        $ \sqrt{n} ( \hat{\bm{\theta}}_{L_0} - \bm{\theta}_0 ) \to^d   N(0, \bm{\Sigma}_{\theta}^{-1} )  $.
	    \end{proof}

	    \begin{proof}[Proof of Proposition~\ref{th:proposition}.]
	         Both results immediately follow from Theorem~2(B) in \citet{shen2013constrained} considering our SFSOD formulation based on the two disjoint sets $\mathcal{S}_\beta$ and $\mathcal{S}_\phi$.
	    \end{proof}
    
    \phantomsection
	\section*{Appendix B: Simulation study} 
    \renewcommand{\thetable}{B.\arabic{table}}
	\setcounter{table}{0}

        \if1\blind
        {
            Our simulations were performed through the
            high-performance computing infrastructure of the
            \textit{Institute for Computational and Data Sciences Advanced CyberInfrastructure} (ICDS-ACI) at Penn State University. We used basic memory option on the ACI-B cluster with an Intel Xeon 24 core processor at 2.2 GHz and 128 GB of RAM. The multi-thread option in \texttt{R} and \texttt{Gurobi} was limited to a maximum of 24 threads.
        } \fi
        
        \if0\blind
        {
            Our simulations were performed through an high-performance computing infrastructure.
            In particular, we used an Intel Xeon 24 core processor at 2.2 GHz and 128 GB of RAM. The multi-thread option in \texttt{R} and \texttt{Gurobi} was limited to a maximum of 24 threads.
        } \fi
        
        We compared the following estimators:
        \begin{enumerate}

            \item[a.] The \textit{sparse-LTS} estimator combines an $L_1$-penalty with the LTS estimator \citep{alfons2013sparse}. 
            Similarly to other methods, we do not perform a final re-weighting step.
            The algorithm starts with 1000 initial subsamples, where 20 subsamples with the lowest value of the objective function are used to compute additional concentration-steps until convergence. 
            The sparsity level is tuned according to the BIC-type criterion proposed by the authors.
            Our implementation is based on the parallelized \texttt{sparseLTS} function of the \texttt{robustHD} package \citep{package_robHD} in \texttt{R} (we use the \texttt{R} version 3.6.1).

            \item[b.] The \textit{enet-LTS} estimator combines an elastic-net penalty with the LTS loss function \citep{kurnaz2017robust}. Also here we use a Lasso penalty and the algorithm starts with 1000 initial subsamples,
            where 20 subsamples with the lowest value of the objective function are used to compute additional concentration-steps until convergence. 
            The sparsity level is tuned through a robust 10-folds cross-validation as advocated by the authors.
            It is implemented using the \texttt{enetLTS} function (without parallelization) of the homonymous
            \texttt{R} package \citep{package_enetLTS}.

            \item[c.] Our \textit{MIP} procedure solves \eqref{eq:miqp} based on the $L_2$-loss and excludes the ridge-like penalty. 
            The sparsity level $k_p^*$ ranges from 1 (only the intercept term) to $2 p_0$ and is tuned through a BIC-type criterion (see also Section~\ref{sec2:proposalTech}). 
            This is computed as:
            $
            \text{BIC}(k_p^*) = k_p^* \log(h) + h \log( L )
            $,
            where $h = n-n_0$ and
            $ L = h^{-1} \norm{ \bm{y} - \bm{A} \hat{\bm{\theta}}_{L_0}}_2^2 $.
            Instead of taking the minimum BIC, we aim at finding an elbow across the considered $k_p^*$ values. 
            As a simple approach, our final solution is the one with the largest absolute decrease along consecutive model sizes, i.e.~$ k_p = \min_{k_p^*}\{ \text{BIC}(k_p^*) - \text{BIC}(k_p^* - 1) \} $. 
            For each $k_p^*$ value, the corresponding MIQP is warm-started using the result from the previous model of size $k_p^* - 1$.
            Big-$\mathcal{M}$ bounds are computed using the ensemble method described in Section~\ref{sec2:proposalTech}, including all estimators used in our comparison (apart from the oracle).
            Our implementation is based on the \texttt{Julia} programming language (version 0.6.0) in connection with the MIP commercial solver \texttt{Gurobi} (version 8.1.1) through the \texttt{JuMP} package. Our code can run in parallel and is provided in Appendix~D.
            Each job runs with a scheduled time limit of 300 seconds.
            
        \end{enumerate}
                
        Table~\ref{tab:1medians} shows our results in terms of medians and MADs for the simulation setting discussed in Section~\ref{sec3:sim}.
        A comparison with Table~\ref{tab:1means} highlights the skewness of most metrics, with all methods performing better, especially our proposal (denoted MIP).
        \begin{table}[!ht]
        \centering
        \caption{Median (MAD in parenthesis) of RMSPE, variance and squared bias for $\hat{\beta}$, FPR and FNR for $\hat{\beta}$ and outlier detection, and computing time, 
        based on 100 simulation replications for the simulation setting in Section~\ref{sec3:sim}.} 
        \label{tab:1medians}
        \begingroup\scriptsize\renewcommand\arraystretch{0.5} 
        \setlength{\tabcolsep}{4pt}
        \begin{tabular}{cclcccccccc}
        \midrule \midrule
        \multicolumn{1}{c}{$n$} & \multicolumn{1}{c}{$p$} & \multicolumn{1}{c}{Method} & \multicolumn{1}{c}{RMSPE} &
        \multicolumn{1}{c}{$\text{var}(\hat{\bm{\beta}})$} &
        \multicolumn{1}{c}{$\text{bias}(\hat{\bm{\beta}})^2$} & 
        \multicolumn{1}{c}{$\text{FPR}(\hat{\bm{\beta}})$} & 
        \multicolumn{1}{c}{$\text{FNR}(\hat{\bm{\beta}})$} & 
        \multicolumn{1}{c}{$\text{FPR}(\hat{\bm{\phi}})$} & 
        \multicolumn{1}{c}{$\text{FNR}(\hat{\bm{\phi}})$} & 
        \multicolumn{1}{c}{Time} \\
        \midrule
        50 & 50 & Oracle & 1.90(0.23) & 0.01(0.00) & 0.00(0.00) & 0.00(0.00) & 0.00(0.00) & 0.00(0.00) & 0.00(0.00) &   0.00(0.00) \\ 
        &  & EnetLTS & 2.35(0.42) & 0.05(0.00) & 0.03(0.00) & 0.11(0.07) & 0.00(0.00) & 0.00(0.00) & 0.00(0.00) &  14.74(0.73) \\ 
        &  & SparseLTS & 2.45(0.48) & 0.06(0.00) & 0.00(0.00) & 0.53(0.07) & 0.00(0.00) & 0.00(0.00) & 0.00(0.00) &   1.55(0.23) \\ 
        &  & MIP & 1.95(0.40) & 0.04(0.00) & 0.00(0.00) & 0.00(0.00) & 0.00(0.00) & 0.00(0.00) & 0.00(0.00) &   8.82(3.85) \\ 
        [2mm]
        100 & 50 & Oracle & 1.86(0.15) & 0.00(0.00) & 0.00(0.00) & 0.00(0.00) & 0.00(0.00) & 0.00(0.00) & 0.00(0.00) &   0.00(0.00) \\ 
        &  & EnetLTS & 2.00(0.18) & 0.01(0.00) & 0.00(0.00) & 0.22(0.20) & 0.00(0.00) & 0.00(0.00) & 0.00(0.00) &  12.21(0.20) \\ 
        &  & SparseLTS & 2.14(0.20) & 0.03(0.00) & 0.00(0.00) & 0.67(0.08) & 0.00(0.00) & 0.00(0.00) & 0.00(0.00) &   2.20(0.38) \\ 
        &  & MIP & 1.86(0.16) & 0.01(0.00) & 0.00(0.00) & 0.00(0.00) & 0.00(0.00) & 0.00(0.00) & 0.00(0.00) &  33.02(10.98) \\ 
        [2mm]
        150 & 50 & Oracle & 1.80(0.13) & 0.00(0.00) & 0.00(0.00) & 0.00(0.00) & 0.00(0.00) & 0.00(0.00) & 0.00(0.00) &   0.00(0.00) \\ 
        &  & EnetLTS & 1.88(0.14) & 0.01(0.00) & 0.00(0.00) & 0.30(0.21) & 0.00(0.00) & 0.00(0.00) & 0.00(0.00) &  12.65(0.24) \\ 
        &  & SparseLTS & 1.96(0.17) & 0.02(0.00) & 0.00(0.00) & 0.67(0.07) & 0.00(0.00) & 0.00(0.00) & 0.00(0.00) &   4.51(0.92) \\ 
        &  & MIP & 1.80(0.14) & 0.00(0.00) & 0.00(0.00) & 0.00(0.00) & 0.00(0.00) & 0.00(0.00) & 0.00(0.00) & 415.14(331.70) \\ 
        [2mm]
        50 & 200 & Oracle & 1.88(0.22) & 0.00(0.00) & 0.00(0.00) & 0.00(0.00) & 0.00(0.00) & 0.00(0.00) & 0.00(0.00) &   0.00(0.00) \\ 
        &  & EnetLTS & 2.86(0.74) & 0.03(0.00) & 0.02(0.00) & 0.17(0.04) & 0.00(0.00) & 0.00(0.00) & 0.00(0.00) &  43.99(2.48) \\ 
        &  & SparseLTS & 2.66(0.51) & 0.02(0.00) & 0.01(0.00) & 0.16(0.02) & 0.00(0.00) & 0.00(0.00) & 0.00(0.00) &   3.37(0.57) \\ 
        &  & MIP & 1.92(0.31) & 0.02(0.00) & 0.00(0.00) & 0.00(0.00) & 0.00(0.00) & 0.00(0.00) & 0.00(0.00) &   6.91(4.22) \\ 
        [2mm]
        100 & 200 & Oracle & 1.80(0.20) & 0.00(0.00) & 0.00(0.00) & 0.00(0.00) & 0.00(0.00) & 0.00(0.00) & 0.00(0.00) &   0.00(0.00) \\ 
        &  & EnetLTS & 2.28(0.45) & 0.02(0.00) & 0.01(0.00) & 0.24(0.09) & 0.00(0.00) & 0.00(0.00) & 0.00(0.00) &  53.96(3.34) \\ 
        &  & SparseLTS & 2.29(0.25) & 0.01(0.00) & 0.00(0.00) & 0.31(0.02) & 0.00(0.00) & 0.00(0.00) & 0.00(0.00) &  11.40(2.48) \\ 
        &  & MIP & 1.80(0.20) & 0.00(0.00) & 0.00(0.00) & 0.00(0.00) & 0.00(0.00) & 0.00(0.00) & 0.00(0.00) &  65.80(27.63) \\ 
        [2mm]
        150 & 200 & Oracle & 1.82(0.16) & 0.00(0.00) & 0.00(0.00) & 0.00(0.00) & 0.00(0.00) & 0.00(0.00) & 0.00(0.00) &   0.00(0.00) \\ 
        &  & EnetLTS & 2.07(0.24) & 0.02(0.00) & 0.01(0.00) & 0.21(0.14) & 0.00(0.00) & 0.01(0.01) & 0.00(0.00) &  56.98(3.41) \\ 
        &  & SparseLTS & 2.25(0.21) & 0.01(0.00) & 0.00(0.00) & 0.42(0.04) & 0.00(0.00) & 0.00(0.00) & 0.00(0.00) &  19.09(3.37) \\ 
        &  & MIP & 1.82(0.16) & 0.00(0.00) & 0.00(0.00) & 0.00(0.00) & 0.00(0.00) & 0.00(0.00) & 0.00(0.00) & 599.00(360.73) \\ 
        \bottomrule 
        \end{tabular}
        \endgroup
        \end{table}
 
        We also explored weak SNR scenarios.
        The following simulation setting is the same as in Section~\ref{sec3:sim}, with the only difference being that the signal-to-noise-ratio is reduced to $\text{SNR}=3$.
        Table~\ref{tab:1SNR3} shows simulation results in term of medians and MADs. A comparison with Table~\ref{tab:1means} shows that similar conclusions hold, although all methods experience an overall decrease in performance.
        Our approach generally outperforms other methods and converges faster to the oracle solution.
        However, for scenarios with small sample sizes, the FNR in $\hat{\bm{\beta}}$ for our method is worse. 
        Moreover, while computing time for heuristic methods remains similar to the stronger SNR scenario, our proposal shows a marked increase.
        \begin{table}[!ht]
        \centering
        \caption{Median (MAD in parenthesis) of RMSPE, variance and squared bias for $\hat{\beta}$, FPR and FNR for $\hat{\beta}$ and outlier detection, and computing time, 
        based on 100 simulation replications
        for a simulation setting similarly to Section~\ref{sec3:sim}
        with $\text{SNR}=3$.} 
        \label{tab:1SNR3}
        \begingroup\scriptsize\renewcommand\arraystretch{0.5} 
        \setlength{\tabcolsep}{4pt}
        \begin{tabular}{cclcccccccc}
         \midrule \midrule
         \multicolumn{1}{c}{$n$} & \multicolumn{1}{c}{$p$} & \multicolumn{1}{c}{Method} & \multicolumn{1}{c}{RMSPE} &
         \multicolumn{1}{c}{$\text{var}(\hat{\bm{\beta}})$} &
         \multicolumn{1}{c}{$\text{bias}(\hat{\bm{\beta}})^2$} & 
         \multicolumn{1}{c}{$\text{FPR}(\hat{\bm{\beta}})$} & 
         \multicolumn{1}{c}{$\text{FNR}(\hat{\bm{\beta}})$} & 
         \multicolumn{1}{c}{$\text{FPR}(\hat{\bm{\phi}})$} & 
         \multicolumn{1}{c}{$\text{FNR}(\hat{\bm{\phi}})$} & 
         \multicolumn{1}{c}{Time} \\
         \midrule
           50 & 50 & Oracle & 2.49(0.34) & 0.02(0.00) & 0.00(0.00) & 0.00(0.00) & 0.00(0.00) & 0.00(0.00) & 0.00(0.00) & 0.00(0.00) \\ 
           &  & EnetLTS & 3.09(0.73) & 0.14(0.00) & 0.06(0.00) & 0.18(0.13) & 0.00(0.00) & 0.00(0.00) & 0.00(0.00) & 14.76(0.81) \\
          &  & SparseLTS & 3.30(0.57) & 0.15(0.00) & 0.01(0.00) & 0.58(0.07) & 0.00(0.00) & 0.00(0.00) & 0.00(0.00) & 2.56(0.72) \\
           &  & MIP & 2.93(0.99) & 0.10(0.00) & 0.02(0.00) & 0.00(0.00) & 0.20(0.30) & 0.00(0.00) & 0.00(0.00) & 16.17(9.01) \\ 
           [2mm]
          100 & 50 & Oracle & 2.42(0.20) & 0.01(0.00) & 0.00(0.00) & 0.00(0.00) & 0.00(0.00) & 0.00(0.00) & 0.00(0.00) & 0.00(0.00) \\ 
           &  & EnetLTS & 2.58(0.26) & 0.02(0.00) & 0.00(0.00) & 0.30(0.16) & 0.00(0.00) & 0.00(0.00) & 0.00(0.00) & 12.30(0.21) \\ 
           &  & SparseLTS & 2.82(0.27) & 0.05(0.00) & 0.00(0.00) & 0.71(0.10) & 0.00(0.00) & 0.00(0.00) & 0.00(0.00) & 2.40(0.26) \\ 
           &  & MIP & 2.48(0.30) & 0.02(0.00) & 0.00(0.00) & 0.00(0.00) & 0.00(0.00) & 0.00(0.00) & 0.00(0.00) & 60.97(38.33) \\ 
           [2mm]
          150 & 50 & Oracle & 2.29(0.16) & 0.00(0.00) & 0.00(0.00) & 0.00(0.00) & 0.00(0.00) & 0.00(0.00) & 0.00(0.00) & 0.00(0.00) \\ 
           &  & EnetLTS & 2.48(0.23) & 0.02(0.00) & 0.00(0.00) & 0.50(0.25) & 0.00(0.00) & 0.00(0.00) & 0.00(0.00) & 12.82(0.18) \\ 
           &  & SparseLTS & 2.57(0.21) & 0.03(0.00) & 0.00(0.00) & 0.76(0.07) & 0.00(0.00) & 0.00(0.00) & 0.00(0.00) & 4.05(1.26) \\ 
           &  & MIP & 2.31(0.21) & 0.01(0.00) & 0.00(0.00) & 0.00(0.00) & 0.00(0.00) & 0.00(0.00) & 0.00(0.00) & 751.73(302.70) \\ 
           [2mm]
          50 & 200 & Oracle & 2.44(0.31) & 0.00(0.00) & 0.00(0.00) & 0.00(0.00) & 0.00(0.00) & 0.00(0.00) & 0.00(0.00) & 0.00(0.00) \\ 
           &  & EnetLTS & 3.93(1.28) & 0.04(0.00) & 0.03(0.00) & 0.18(0.05) & 0.00(0.00) & 0.01(0.02) & 0.00(0.00) & 45.94(2.43) \\ 
           &  & SparseLTS & 3.67(1.06) & 0.04(0.00) & 0.03(0.00) & 0.17(0.02) & 0.00(0.00) & 0.00(0.00) & 0.00(0.00) & 4.22(1.19) \\ 
           &  & MIP & 3.30(1.51) & 0.05(0.00) & 0.01(0.00) & 0.00(0.00) & 0.20(0.30) & 0.00(0.00) & 0.00(0.00) & 18.76(17.94) \\ 
           [2mm]
          100 & 200 & Oracle & 2.40(0.27) & 0.00(0.00) & 0.00(0.00) & 0.00(0.00) & 0.00(0.00) & 0.00(0.00) & 0.00(0.00) & 0.00(0.00) \\ 
           &  & EnetLTS & 3.18(0.65) & 0.03(0.00) & 0.01(0.00) & 0.30(0.10) & 0.00(0.00) & 0.01(0.02) & 0.00(0.00) & 55.32(3.40) \\ 
           &  & SparseLTS & 3.09(0.30) & 0.02(0.00) & 0.00(0.00) & 0.33(0.02) & 0.00(0.00) & 0.00(0.00) & 0.00(0.00) & 13.58(4.58) \\ 
           &  & MIP & 2.46(0.41) & 0.01(0.00) & 0.00(0.00) & 0.00(0.00) & 0.00(0.00) & 0.00(0.00) & 0.00(0.00) & 107.57(74.31) \\ 
           [2mm]
          150 & 200 & Oracle & 2.37(0.21) & 0.00(0.00) & 0.00(0.00) & 0.00(0.00) & 0.00(0.00) & 0.00(0.00) & 0.00(0.00) & 0.00(0.00) \\ 
           &  & EnetLTS & 2.76(0.31) & 0.02(0.00) & 0.01(0.00) & 0.28(0.14) & 0.00(0.00) & 0.00(0.00) & 0.00(0.00) & 56.74(3.69) \\ 
           &  & SparseLTS & 2.98(0.34) & 0.02(0.00) & 0.00(0.00) & 0.45(0.02) & 0.00(0.00) & 0.00(0.00) & 0.00(0.00) & 12.60(0.89) \\ 
           &  & MIP & 2.39(0.22) & 0.00(0.00) & 0.00(0.00) & 0.00(0.00) & 0.00(0.00) & 0.00(0.00) & 0.00(0.00) & 840.34(519.30) \\ 
           \bottomrule 
        \end{tabular}
        \endgroup
        \end{table}

        We also explored simulation settings with multicollinearity structures.
        Table~\ref{tab:1corr} presents our results for a simulation scenario which mimics that of Section~\ref{sec3:sim} (reporting medians and MADs), with the only difference being that $ \bm{\Sigma}_X $ has an autoregressive correlation structure $  \Sigma_{X, ij } = 0.3^{\lvert i - j \rvert} $. Though this could be considered a ``mild'' level of correlation, we note that the addition of contamination increases the amount of multicollinearity present.
        Here our approach is often outperformed by other methods for small sample sizes,
        however as the latter increase we can again notice that our proposal converges faster to the oracle solution and results like those in Table~\ref{tab:1means} hold.
        
        Smaller SNR regimes and stronger correlation structures were also explored; results are not reported as all methods performed quite poorly. In these settings, as advocated in \citet{hastie2017extended}, a ridge-like penalty may be beneficial.
        
        \begin{table}[!ht]
        \centering
        \caption{Median (MAD in parenthesis) of RMSPE, variance and squared bias for $\hat{\beta}$, FPR and FNR for $\hat{\beta}$ and outlier detection, and computing time, 
        based on 100 simulation replications 
        for a simulation setting similarly to Section~\ref{sec3:sim}
        in presence of multicollinearity.} 
        \label{tab:1corr}
        \begingroup\scriptsize\renewcommand\arraystretch{0.5} 
        \setlength{\tabcolsep}{4pt}
        \begin{tabular}{cclcccccccc}
         \midrule \midrule
         \multicolumn{1}{c}{$n$} & \multicolumn{1}{c}{$p$} & \multicolumn{1}{c}{Method} & \multicolumn{1}{c}{RMSPE} &
         \multicolumn{1}{c}{$\text{var}(\hat{\bm{\beta}})$} &
         \multicolumn{1}{c}{$\text{bias}(\hat{\bm{\beta}})^2$} & 
         \multicolumn{1}{c}{$\text{FPR}(\hat{\bm{\beta}})$} & 
         \multicolumn{1}{c}{$\text{FNR}(\hat{\bm{\beta}})$} & 
         \multicolumn{1}{c}{$\text{FPR}(\hat{\bm{\phi}})$} & 
         \multicolumn{1}{c}{$\text{FNR}(\hat{\bm{\phi}})$} & 
         \multicolumn{1}{c}{Time} \\
         \midrule
          50 & 50 & Oracle & 2.34(0.26) & 0.02(0.00) & 0.00(0.00) & 0.00(0.00) & 0.00(0.00) & 0.00(0.00) & 0.00(0.00) & 0.00(0.00) \\ 
           &  & EnetLTS & 2.75(0.36) & 0.05(0.00) & 0.02(0.00) & 0.07(0.07) & 0.00(0.00) & 0.00(0.00) & 0.00(0.00) & 14.71(0.55) \\ 
         &  & SparseLTS & 2.92(0.49) & 0.08(0.00) & 0.00(0.00) & 0.51(0.07) & 0.00(0.00) & 0.00(0.00) & 0.00(0.00) & 1.39(0.22) \\ 
           &  & MIP & 3.00(0.59) & 0.09(0.00) & 0.00(0.00) & 0.11(0.03) & 0.00(0.00) & 0.00(0.00) & 0.00(0.00) & 9.05(2.67) \\ 
           [2mm]
          100 & 50 & Oracle & 2.30(0.21) & 0.01(0.00) & 0.00(0.00) & 0.00(0.00) & 0.00(0.00) & 0.00(0.00) & 0.00(0.00) & 0.00(0.00) \\ 
           &  & EnetLTS & 2.43(0.23) & 0.01(0.00) & 0.01(0.00) & 0.09(0.10) & 0.00(0.00) & 0.00(0.00) & 0.00(0.00) & 12.63(0.18) \\ 
           &  & SparseLTS & 2.57(0.25) & 0.03(0.00) & 0.00(0.00) & 0.56(0.10) & 0.00(0.00) & 0.00(0.00) & 0.00(0.00) & 1.77(0.20) \\ 
           &  & MIP & 2.46(0.26) & 0.02(0.00) & 0.00(0.00) & 0.04(0.03) & 0.00(0.00) & 0.00(0.00) & 0.00(0.00) & 26.37(8.51) \\ 
           [2mm]
          150 & 50 & Oracle & 2.21(0.17) & 0.00(0.00) & 0.00(0.00) & 0.00(0.00) & 0.00(0.00) & 0.00(0.00) & 0.00(0.00) & 0.00(0.00) \\ 
           &  & EnetLTS & 2.36(0.17) & 0.01(0.00) & 0.00(0.00) & 0.18(0.20) & 0.00(0.00) & 0.00(0.00) & 0.00(0.00) & 13.13(0.16) \\ 
           &  & SparseLTS & 2.38(0.19) & 0.02(0.00) & 0.00(0.00) & 0.54(0.05) & 0.00(0.00) & 0.00(0.00) & 0.00(0.00) & 3.23(0.82) \\ 
           &  & MIP & 2.31(0.20) & 0.01(0.00) & 0.00(0.00) & 0.02(0.03) & 0.00(0.00) & 0.00(0.00) & 0.00(0.00) & 340.60(175.80) \\ 
           [2mm]
          50 & 200 & Oracle & 2.33(0.33) & 0.00(0.00) & 0.00(0.00) & 0.00(0.00) & 0.00(0.00) & 0.00(0.00) & 0.00(0.00) & 0.00(0.00) \\ 
           &  & EnetLTS & 3.19(0.96) & 0.04(0.00) & 0.01(0.00) & 0.16(0.05) & 0.00(0.00) & 0.00(0.00) & 0.00(0.00) & 44.68(2.70) \\ 
           &  & SparseLTS & 3.01(0.55) & 0.02(0.00) & 0.00(0.00) & 0.16(0.02) & 0.00(0.00) & 0.00(0.00) & 0.00(0.00) & 3.48(0.82) \\ 
           &  & MIP & 3.01(0.74) & 0.03(0.00) & 0.00(0.00) & 0.02(0.01) & 0.00(0.00) & 0.00(0.00) & 0.00(0.00) & 8.30(5.25) \\ 
           [2mm]
          100 & 200 & Oracle & 2.27(0.23) & 0.00(0.00) & 0.00(0.00) & 0.00(0.00) & 0.00(0.00) & 0.00(0.00) & 0.00(0.00) & 0.00(0.00) \\ 
           &  & EnetLTS & 2.63(0.32) & 0.02(0.00) & 0.00(0.00) & 0.18(0.12) & 0.00(0.00) & 0.00(0.00) & 0.00(0.00) & 54.46(2.21) \\ 
           &  & SparseLTS & 2.76(0.32) & 0.01(0.00) & 0.00(0.00) & 0.29(0.02) & 0.00(0.00) & 0.00(0.00) & 0.00(0.00) & 8.74(1.75) \\ 
           &  & MIP & 2.75(0.29) & 0.01(0.00) & 0.00(0.00) & 0.03(0.00) & 0.00(0.00) & 0.00(0.00) & 0.00(0.00) & 66.16(20.62) \\ 
           [2mm]
          150 & 200 & Oracle & 2.24(0.20) & 0.00(0.00) & 0.00(0.00) & 0.00(0.00) & 0.00(0.00) & 0.00(0.00) & 0.00(0.00) & 0.00(0.00) \\ 
           &  & EnetLTS & 2.46(0.24) & 0.01(0.00) & 0.00(0.00) & 0.16(0.10) & 0.00(0.00) & 0.00(0.00) & 0.00(0.00) & 59.83(2.19) \\ 
           &  & SparseLTS & 2.67(0.23) & 0.01(0.00) & 0.00(0.00) & 0.35(0.05) & 0.00(0.00) & 0.00(0.01) & 0.00(0.00) & 14.67(2.96) \\ 
           &  & MIP & 2.55(0.21) & 0.01(0.00) & 0.00(0.00) & 0.03(0.00) & 0.00(0.00) & 0.00(0.00) & 0.00(0.00) & 375.21(247.13) \\ 
           \bottomrule 
        \end{tabular}
        \endgroup
        \end{table}

        \phantomsection
    	\section*{Appendix C: Application study}
    	\renewcommand{\theequation}{C.\arabic{equation}}
    	\setcounter{equation}{0}
    	\renewcommand{\thetable}{C.\arabic{table}}
    	\setcounter{table}{0}
    
        	 Our analysis of the real datasets was performed through the same high-performance computing infrastructure as our simulations. We used an Intel Xeon 24 core processor at 2.2 GHz and 128 GB of RAM. The multi-thread option in \texttt{R} and \texttt{Gurobi} was limited to a maximum of 24 threads.
        	 
        	 Since the test set may contain any number of outliers, we utilized a (robustly) scaled estimation procedure as proposed in \citet{mishra2019robust} for our comparisons in Table~\ref{tab:2}. After tuning $k_p$, we compute the training and test residuals ($\bm{\epsilon}_{\text{tr}}$ and $\bm{\epsilon}_{\text{te}}$), training and test scaling parameter ($ s_{\text{tr}}$ and $ s_{\text{te}}$), and scaled test error ($\bm{r}_{\text{te}}$) as follows:
        	 \begin{align}
        	     \bm{\epsilon}_{\text{tr}} &= \bm{y}_{\text{tr}}-\bm{X}_{\text{tr}}\hat{\bm{\beta}}-\hat{\bm{\phi}} \nonumber \\
        	     s_{\text{tr}} &= \sqrt{||\bm{\epsilon}_{\text{tr}}||^2_2 /n_{\text{tr}}} \nonumber \\
        	     \bm{\epsilon}_{\text{te}} &= \bm{y}_{\text{te}}-\bm{X}_{\text{te}}\hat{\bm{\beta}} \nonumber \\
        	     s_{\text{te}} &= \text{MAD}(\bm{r}_{\text{te}}) \nonumber \\
        	     \bm{r}_{\text{te}} &=\bm{\epsilon}_{\text{te}}/\bm{s}_{\text{tr}} \nonumber .
        	 \end{align}
        	 Outliers from $\bm{r}_{\text{te}}$ are identified and removed based on the threshold rule $|\bm{r}_{\text{te}}|>1.345s_{\text{te}}$ (as specified in \citealt{mishra2019robust}) and the average SPE is taken over the 
        	 estimated non-outlying points.
        	 
        	   Lastly, we provide the median results across replications in Table~\ref{tab:2medians} corresponding to the same instances provided in Table~2 of the main text. We compare SPE values across robust approaches, enet-LTS, sparse-LTS and our MIP proposal (described in Appendix~B) as well as a classic Lasso. More specifically, we use a Lasso penalty and tune via 10-fold cross-validation across a grid of at most 100 tuning parameters. It is implemented using the \texttt{cv.glmnet} function within the \texttt{glmnet} \texttt{R} package \citep{friedman2010}.
        	  \begin{table}[!ht]
                \centering
                \caption{Median (MAD in parenthesis) of SPE, estimated nonzero coefficients on training data, and estimated number of non-outlying points on test data, based on 5 replications,
                and estimated nonzero coefficients on the full data, 
                for the human microbiome analysis using 10\% trimming for robust methods.}
                \label{tab:2medians}
                \begingroup\scriptsize\renewcommand\arraystretch{0.5} 
                \setlength{\tabcolsep}{12.2pt}
                \begin{tabular}{lccclcccc}
                 \midrule \midrule
                 \multicolumn{1}{c}{Data} & \multicolumn{1}{c}{$n^{\text{tr}}$} & \multicolumn{1}{c}{$n^{\text{te}}$} & \multicolumn{1}{c}{$p$} &
                 \multicolumn{1}{c}{Method} &
                 \multicolumn{1}{c}{SPE} & 
                 \multicolumn{1}{c}{$\hat{p}_0^{\text{tr}}$} & 
                 \multicolumn{1}{c}{$n^{\text{te}}-\hat{n}_{0}^{\text{te}} $}  &
                 \multicolumn{1}{c}{$\hat{p}_0^{\text{full}}$}\\
                 \midrule
                 Child oral & 172 & 43 & 78 & EnetLTS & 0.41(0.11) & 58(0.86) & 35(1.79) & 60\\ 
                 & &  &  & SparseLTS & 0.43(0.08) & 56(6.2) & 33(1.50) & 55\\ 
                 & & &  & MIP & 0.32(0.04) & 2(0.8) & 34(1.33) & 2\\ 
                 & &  &  & Lasso & 0.51(0.05) & 1(2.14) & 33(1.60) & 10\\
                   [2mm]
                 Child gut & 152 & 38 & 75 & EnetLTS & 0.33(0.13) & 58(0.24) & 29(1.17) & 54\\ 
                 & &  &  & SparseLTS & 0.34(0.02) & 54(4.69) & 30(1.30) & 74\\ 
                 & & &        & MIP & 0.31(0.04) & 2(1.26) & 30(0.97) & 2\\ 
                 & &  &  & Lasso & 0.68(0.13) & 1(1.05) & 33(1.54) & 1\\
                   [2mm]
                 Maternal oral &  172 & 43 & 79 & EnetLTS & 0.45(0.03) & 64(1.63) & 34(1.07) & 65\\ 
                 & &  &  & SparseLTS & 0.36(0.14) & 58(5.32) & 34(1.02) & 60\\ 
                 & & &  & MIP & 0.38(0.09) & 3(1.05) & 37(0.80) & 6\\ 
                 & &  &  & Lasso & 0.78(0.09) & 1(0.00) & 36(0.97) & 1\\
                   \bottomrule 
                \end{tabular}
                \endgroup
            \end{table}
 
        \phantomsection
    	\section*{Appendix D: Code}
        	
        	Our code is available upon request.

\begingroup
    {\small
    \bibliography{biblio.bib}
    }
\endgroup

\end{document}